\documentclass[11pt]{article}
\usepackage{palatino}
\usepackage{url}

\usepackage{pgfplots} 
\usepackage{tikz, tikzpeople}
\usetikzlibrary{calc,intersections,decorations.pathmorphing,decorations.markings,patterns,arrows.meta,positioning}
\usetikzlibrary{angles,quotes}
\usetikzlibrary{arrows.meta}
\usetikzlibrary{arrows}
\usetikzlibrary{shapes.multipart}
\usetikzlibrary{fit}
\usetikzlibrary{shapes.geometric,positioning}
\usetikzlibrary{calc}
\usepackage{bm}
\usepackage{mathtools}
\usepackage{xfrac}
\usepackage{graphicx}
\usepackage[linesnumbered,ruled,vlined]{algorithm2e}
\usepackage{cite}
\usepackage{amsmath,amssymb,amsfonts, amsthm,  mathrsfs}
\usepackage{algpseudocode}
\usepackage{textcomp}
\usepackage{xcolor}
\usepackage{braket}
\usepackage{enumitem}
\usepackage{cancel}
\usepackage{stfloats}

\usetikzlibrary{patterns}
\newtheorem{lemma}{Lemma}
\newtheorem{theorem}{Theorem}
\newtheorem{proposition}{Proposition}

\newtheorem{remark}{Remark}

\newtheorem{definition}{Definition}

\newlength{\leftstackrelawd}
\newlength{\leftstackrelbwd}
\def\leftstackrel#1#2{\settowidth{\leftstackrelawd}%
{${{}^{#1}}$}\settowidth{\leftstackrelbwd}{$#2$}%
\addtolength{\leftstackrelawd}{-\leftstackrelbwd}%
\leavevmode\ifthenelse{\lengthtest{\leftstackrelawd>0pt}}%
{\kern-.5\leftstackrelawd}{}\mathrel{\mathop{#2}\limits^{#1}}}

\DeclareMathOperator{\Tr}{Tr}

\newcommand{\mN}{\mathsf{N}} 
 
\newcommand{\mP}{\mathsf{P}}

\newcommand{\ch}{\mathsf{ch}}
\newcommand{\gm}{\mathsf{gm}}

\newcommand{\C}{\mathchoice
  {\mathrm{\scriptscriptstyle C}} % displaystyle
  {\mathrm{\scriptscriptstyle C}} % textstyle
  {\mathrm{\scriptscriptstyle C}} % scriptstyle
  {\mathrm{\scriptscriptstyle C}} % scriptscriptstyle
}

\newcommand{\EA}{\mathchoice
  {\mathrm{\scriptscriptstyle EA}} % displaystyle
  {\mathrm{\scriptscriptstyle EA}} % textstyle
  {\mathrm{\scriptscriptstyle EA}} % scriptstyle
  {\mathrm{\scriptscriptstyle EA}} % scriptscriptstyle
}

\newcommand{\NS}{\mathchoice
  {\mathrm{\scriptscriptstyle NS}} % displaystyle
  {\mathrm{\scriptscriptstyle NS}} % textstyle
  {\mathrm{\scriptscriptstyle NS}} % scriptstyle
  {\mathrm{\scriptscriptstyle NS}} % scriptscriptstyle
}

\allowdisplaybreaks

\title{Unlocking  Exponential and Unbounded Robust Gains in \\Shannon Capacity of Classical Multiple Access Channels \\with Causal CSIT via Quantum Entanglement Assistance}

\author{Yuhang Yao, Syed A. Jafar\\
{\small Nhu Department of Electrical Engineering and Computer Science}\\
{\small University of California Irvine, Irvine, CA 92697}\\
{\small \it Email: \{yuhangy5, syed\}@uci.edu}
}

\date{}

\begin{document}
\maketitle

\begin{abstract}
Quantum entanglement assistance is known to improve the Shannon capacity of classical communication networks but the largest gains noted thus far are rather modest (less than $6\%$), motivating the question: are  \emph{large} capacity gains ever possible? It is shown in this work that in the presence of causal channel state information at the transmitters (CSIT), quantum entanglement assistance provides a multiplicative capacity advantage that grows exponentially with the number of users $K$  for certain classical $K$-user multiple access channels with fixed size (binary) alphabet for inputs, outputs and states. Similarly, in the presence of causal channel state information at the transmitters, quantum entanglement assistance is shown to provide a multiplicative capacity advantage that is unbounded as the size of the state alphabet grows, while the number of users $(K=3)$ and the input and output alphabet (binary) are held fixed. Even with only a few users and small alphabet sizes, substantial multiplicative gains in capacity are found, e.g., with binary inputs, outputs and states, multiplicative gains by factors exceeding 21 and 88 are noted with $K=5$ and $K=7$ users, respectively. The gains are robust in the sense that they persist even with noisy quantum resources, e.g., an exponential (in $K$) capacity advantage from quantum entanglement assistance remains available even if each entangled qubit independently depolarizes completely with probability $\approx$ 30$\%$. The gains are based on quantum entanglement assistance provided only to the transmitters.
\end{abstract}
\newpage

\section{Introduction}
With continued progress toward a quantum internet \cite{Qinternet,Chehimi_Saad,caleffi_tutorial2} likely to make quantum resources increasingly accessible, it is important to understand how  these technologies might transform modern society. One promising avenue lies in using distributed quantum resources to substantially enhance the capacity of \emph{classical} communication networks.

Quantum entanglement is a resource  that can be shared among the communication nodes in advance, i.e., prior to the beginning of communication. It is a static (non-signaling) resource \cite{Wilde_CQE}, i.e., it does not by itself allow any communication. However, it expands the scope of encoding and decoding schemes that can be employed over the classical channel. The transmitting nodes may perform input-dependent local measurements on their respective quantum systems, and determine the symbols to be transmitted on the classical channel based on the classical outcomes of those measurements. Entangled quantum systems can generate \emph{non-local} correlations in the measurement outcomes.  
These are correlations that cannot be generated using shared randomness and local operations alone, without additional communication.
It is natural to explore what capacity gains, if any, are possible from such non-local correlations, i.e., from quantum entanglement assistance, in classical communication networks. 

Our focus is exclusively on how the Shannon (sum-rate) capacity of classical communication networks --- which assumes that all messages and channels are classical and requires asymptotically vanishing probability of error for code lengths approaching infinity --- may be improved by quantum entanglement assistance. Notably, this excludes notions like zero-error capacity \cite{Shannon56}, finite-blocklength capacity \cite{Polyanskiy_Poor_Verdu}, and quantum-channel capacity \cite{Devetak2005}. Significant advantages from quantum-entanglement assistance are long established in the latter categories, but thus far quite rare and modest within the former setting that is our focus in this paper. For example, it is known for well over a decade that quantum-assistance can significantly increase both one-shot and asymptotic zero-error capacities \cite{cubitt2010improving, cubitt2011zero, leung2012entanglement, Piovesan2015}. For quantum channels, it is known for over three decades that entanglement assistance between the parties can significantly improve the capacity. For example, superdense coding \cite{Superdense} enables the transmission of $2$ bits of classical information by sending one noiseless qubit and consuming one pre-shared entangled qubit pair between the transmitter and the receiver, achieving a factor of $2$ gain over the corresponding capacity  without entanglement assistance. For noisy quantum channels, unbounded capacity gain from entanglement assistance is known for various point-to-point channels  \cite{Bennett_Shor_Smolin_Thapliyal_PRL, holevo2004entanglement, guha2020infinite}. For example, Bennett et al.  \cite{Bennett_Shor_Smolin_Thapliyal_PRL, bennett_shor_capacity} and King \cite{king2003capacity}, establish that the classical communication capacity of the $d$-ary quantum depolarizing channel can be up to $d+1$ times larger  than the capacity without entanglement assistance. 

In contrast, until relatively recently, there was little evidence of  quantum-assisted Shannon-capacity advantages for classical communication over classical channels. Early results were negative. It was shown in 1999 by Bennett et al. \cite{Bennett_Shor_Smolin_Thapliyal_PRL} that in a point-to-point channel there is no capacity gain from quantum entanglement assistance. The negative result was further strengthened in 2012 by Matthews \cite{matthews2012linear}, proving that even nonsignaling assistance (which includes quantum assistance as a special case) cannot increase the capacity of a point-to-point channel. In 2014, Beigi and Gohari \cite{Beigi_Gohari_2014} pose the question, `\emph{is there an information theoretic classical setting in which shared entanglement helps}', where their definition of information-theoretic classical settings excludes zero-error, finite-blocklength and quantum channels. Positive results in multiuser settings started to emerge with the 2017 work of Quek and Shor \cite{Quek_Shor, Quek_Dissertation} followed by substantial progress within the past decade \cite{leditzky2020playing, seshadri2023separation, pereg2024MAC_QEassist, fawzi2024MAC,QMAC_Yun, Qcoop_Nam,Yao_Jafar_QEACWS,fawzi2024broadcast,Yao_Jafar_NS_DoF}. Yet, the scope of Shannon capacity advantages from quantum entanglement assistance for classical communication over classical channels remains far from understood. The question is important as future communication infrastructure may remain largely classical, even as it is augmented by quantum resources such as shared entanglement \cite{Qinternet}.

The picture that has emerged thus far is not very optimistic. Among multiuser settings, perhaps the   multiple access channel (MAC) has received the most attention \cite{ leditzky2020playing, seshadri2023separation, pereg2024MAC_QEassist, fawzi2024MAC,QMAC_Yun}. 
Although capacity gains from quantum entanglement assistance have been shown in MACs by leveraging pseudotelepathy games \cite{Quek_Shor, leditzky2020playing}, the magnitude of the gains found thus far is rather modest, not exceeding $6\%$ \cite{seshadri2023separation}.  
Beyond the MAC setting the understanding is even more limited. For example, no capacity gain due to quantum entanglement assistance has yet been found for classical broadcast channels \cite{pereg2021quantumBC,fawzi2024broadcast,Yao_Jafar_NS_DoF}, and it is known that receiver-side entanglement alone cannot improve the capacity \cite{pereg2021quantumBC}. This  raises the question: can quantum entanglement assistance ever produce \emph{large}  gains in the Shannon capacity of classical communication networks?

To address this question, we focus on entanglement-assisted multiple access channels \emph{with causal channel state information at the transmitters} (CSIT). Causal and non-causal CSIT are canonical communication state-dependent frameworks for which the capacity in the point-to-point setting was found by Shannon in \cite{shannon_CSIT} and Gelfand and Pinsker in \cite{gel1980coding}, respectively. Both causal and non-causal CSIT have also been studied in the context of quantum point-to-point channels as well \cite{dupuis_QCSIT, pereg_QCSIT}. As previously noted our focus is on classical channels. The motivation for studying this setting comes from our recent work in \cite{Yao_Jafar_QEACWS}, which explores  the capacity of  point-to-point channels with causal CSIT, and finds examples with capacity gains of up to $12/\pi^2\approx 21\%$, due to quantum entanglement assistance.  The presence of causal CSIT thus appears to be a key that can unlock capacity advantages from quantum entanglement assistance even in scenarios where no such advantage exists otherwise \cite{Bennett_Shor_Smolin_Thapliyal_PRL}. Moreover, causal CSIT is itself a well-motivated assumption in many practical settings, such as wireless uplinks and cognitive radio applications, where  transmitters possess partial knowledge of the channel state based on their local interference environment \cite{Gesbert_dist_CSIT}.

\begin{figure}[htbp]
\center
\begin{tikzpicture}
	\node at (0,0) {\includegraphics[width=0.65\textwidth]{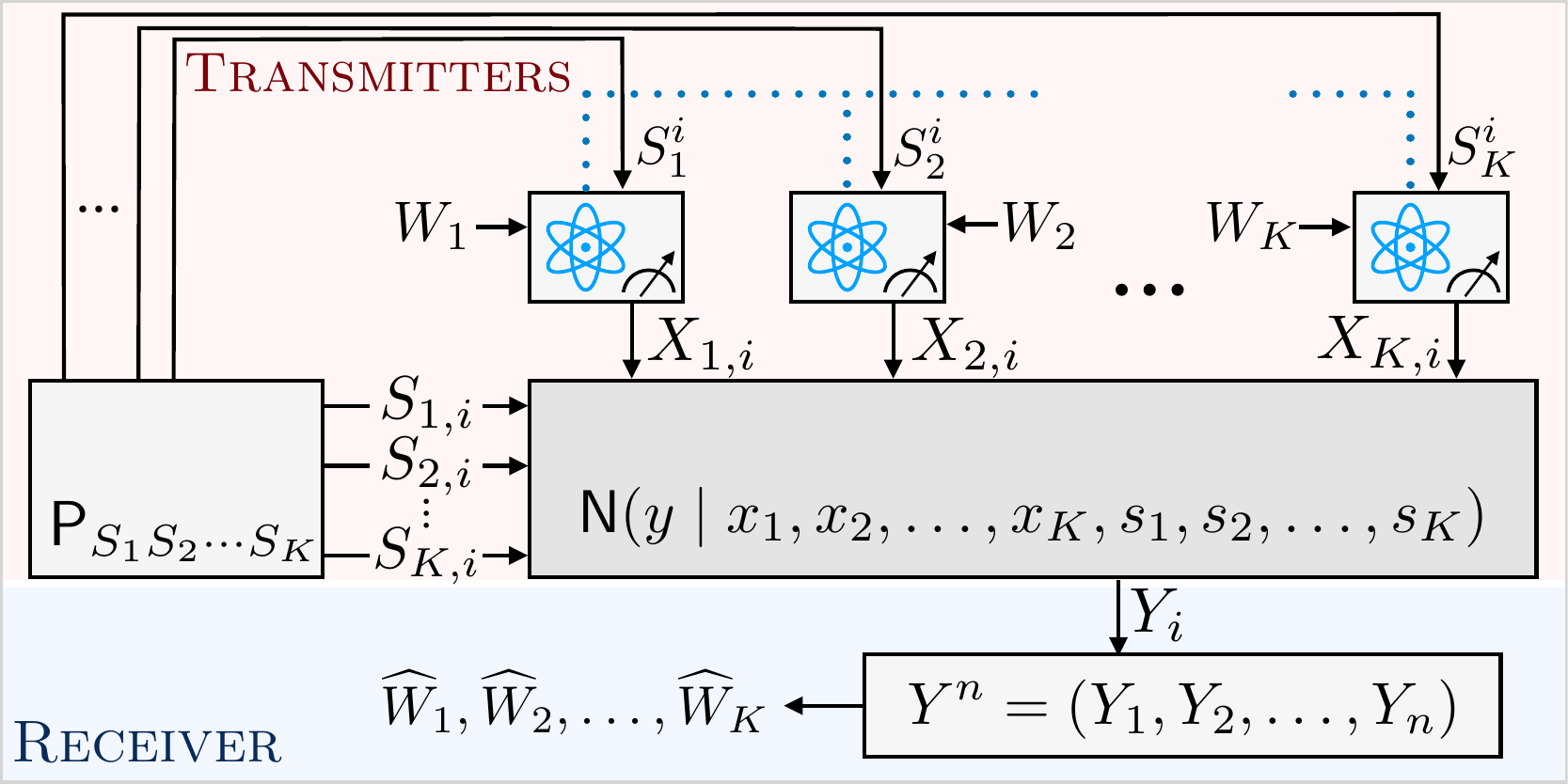}};
	\node[align=center, text={rgb,255:red,43;green,118;blue,182}] at (2.6,2) {Entangled \\ systems};
	\node[align=center] at (-4.15,-0.35) {Channel \\[-2pt] states};
	\node[align=center] at (1.7,-0.3) {State-dependent classical MAC};
\end{tikzpicture}
\caption{Transmitter-side quantum entanglement-assisted classical MAC with causal CSIT. For $k\in \{1,2,\ldots, K\}$, $W_k$ is an independent message that originates at Transmitter $k$. The transmitters have access to pre-shared  quantum entanglement. For the $i^{th}$ use of the channel, Transmitter $k$ conducts a measurement which depends on $W_k$ and the channel states $S_{k}^i = (S_{k,1},S_{k,2},\ldots, S_{k,i})$, and sends the classical measurement outcome $X_{k,i}$ to the channel. After $n$ channel uses, the receiver collects $Y^n=(Y_1,Y_2,\ldots, Y_n)$, from which it decodes the $K$ messages as $(\widehat{W}_1,\widehat{W}_2,\ldots, \widehat{W}_K)$. A formal description of the coding scheme is presented in Section \ref{sec:TxEAcoding}.}
\label{fig:QEAcoding}
\end{figure}

The main contribution of this work is the discovery of \emph{exponential}, 
\emph{unbounded} and \emph{robust}  gains in capacity for various classical $K$-user multiple access channels with causal CSIT, due to quantum entanglement assistance. It is shown (Proposition \ref{prop:ChannelB1}) that in the presence of causal CSIT, quantum entanglement assistance provides a multiplicative capacity advantage that grows exponentially with the number of users $K$  for certain classical $K$-user multiple access channels with fixed size (binary) alphabet for inputs, outputs and states. Similarly, for certain multiple access channels, in the presence of causal CSIT, quantum entanglement assistance is shown (Proposition \ref{prop:ClassC}) to provide a multiplicative capacity advantage that is unbounded as the state alphabet size grows, while the number of users $(K=3)$ and the input and output alphabet (binary) are held fixed.
 
For these gains in capacity, it is sufficient to have entanglement assistance only at the transmitters, as illustrated in Figure \ref{fig:QEAcoding}. For example, one of the channels that we consider (see Definition \ref{def:MerminGHZ_channel}) is a classical binary additive MAC with a state-dependent additive interference term,
\begin{align}
Y=X_1\oplus X_2\oplus\dots\oplus X_K\oplus Z(S_1,S_2,\ldots,S_K).\label{eq:intmit}
\end{align}  Here $S_k$ and $X_k$ denote, respectively, the channel state information symbol and channel input symbol  corresponding to Transmitter $k$. $Y$ is the channel output,  $Z(S_1,S_2,\ldots, S_K)$ is the state-dependent interference observed at the receiver, and $\oplus$ is binary addition (all symbols are binary). This is a generalization,  from a point-to-point setting to a MAC, of the channel model of \cite{Erez_and_Zamir}. Similar models have also been formulated in the studies of wireless networks \cite{Gesbert_dist_CSIT}. For the particular interference structure chosen in our example (inspired by the Mermin-GHZ game \cite{merminGHZ}), we show that the sum-capacity with entanglement assistance provided only to the transmitters is $1$ (bit/channel-use), whereas the capacity without entanglement assistance is only $O(2^{-K})$, decaying exponentially in the number of users $K$.

Given the claims of exponential or unbounded gains, it is natural to ask whether this advantage remains significant in finite-sized regimes, or if constant pre-factors make it practically inaccessible. The answer turns out to be surprisingly positive. Even with a moderate number of users and small alphabet sizes, the gains are substantial. For instance, in the setting of \eqref{eq:intmit}, we observe (Figure \ref{fig:gain}) multiplicative capacity gains exceeding factors of $21$ and $88$ for $K = 5$ and $K = 7$ users, respectively. 

An intuitive insight into the large quantum advantage can be found in the challenge of \emph{distributed interference suppression} that is apparent in  problem formulation \eqref{eq:intmit} through the inclusion of $Z$. Efficient channel utilization requires mitigating interference $Z$ as much as possible based on the available CSIT, while also simultaneously using the channel for communication. How well the interference can be mitigated by distributed transmitters depends on how well the structure of the channel matches the structure of the interference (i.e., how the channel state depends on the CSIT). For example, with a  linear channel as in \eqref{eq:intmit}, a linear interference term may be efficiently neutralized, but a non-linear $Z$ can be much more challenging to suppress for classical coding schemes.
 On the other hand, it is well-known that quantum entanglement assistance, utilizing non-local correlations, can efficiently map certain non-linear distributed computations to linear ones \cite{clauser1969proposed,ambainis2012quantum,ambainis2012advantage,ambainis2013provable,Briet_and_Vidick}. This allows the quantum-assisted transmitters to suppress the non-linear interference $Z$ to a greater extent, with much greater efficiency, while simultaneously communicating (in superposition) over the linear channel, all of which significantly amplifies the capacity advantage. It is also worth noting that while the magnitude of the gain depends on the structure of interference $Z$, even generic structures are found to retain significant advantages from quantum entanglement assistance (see Propositions \ref{prop:ClassA} and \ref{prop:ClassB}).
 
To further understand the significance of CSIT, consider the alternative --- a discrete memoryless MAC  \emph{without} CSIT. Since any channel state information at the receiver can be simply treated as a channel output, a MAC without CSIT can be reduced to a MAC \emph{without state}. A trivial classical coding scheme for such a channel is a time-sharing of $K$ schemes, each optimized for a particular user, so that the sum-rate of the scheme is $1/K$ times the sum of the individual maximum rates achievable by the $K$ users. Recalling that quantum entanglement assistance cannot increase the capacity of a single-user channel \cite{Bennett_Shor_Smolin_Thapliyal_PRL}, it is perhaps intuitively to be expected that quantum-assisted coding schemes may not achieve a multiplicative capacity advantage greater than a factor of $K$. Indeed, this intuition turns out to be correct. We show (Theorem \ref{thm:MAC_no_state} in Section \ref{sec:MAC_no_state}) that in any $K$-user classical multiple access channel model \emph{without state}, a capacity improvement by a multiplicative factor \emph{greater than $K$ is impossible} from quantum entanglement assistance, even if the assistance is provided to all transmitters \emph{and the receiver} --- in fact an improvement by a factor greater than $K$ is impossible even with \emph{non-signaling} assistance (which subsumes quantum assistance as a special case) provided to all parties. Considering that the capacity improvements found thus far (for MAC without state) have been rather small (not exceeding $6\%$ to our knowledge), a factor of $K$ bound could be quite loose. However, even this potentially loose bound clearly rules out any exponential/unbounded gains for multiple access channel models without state. Thus, the presence of a channel state along with some CSIT, exemplified  in \eqref{eq:intmit} by the interference $Z(S_1,S_2,\cdots, S_K)$ with $S_k$ known to Transmitter $k$, emerges as an essential enabling premise for exponential/unbounded gains in capacity from quantum entanglement assistance.

Finally, another natural question to ask is whether these large gains are robust --- do they only materialize if all quantum systems and quantum operations are \emph{perfect}, or do the gains persist even with noisy systems and operations? Here also, the answer is quite positive. The gains persist even with noisy quantum resources. For example,  the entanglement-assisted coding scheme for channel model \eqref{eq:intmit} which achieves an exponential capacity advantage in the number of users $K$, consumes a $K$ qubit (one qubit per transmitter) pre-shared GHZ state per channel use. We find (see Remark \ref{rem:robust} and Appendix \ref{proof:robust}) that the exponential advantage persists  even if each  qubit independently depolarizes  with probability $\approx$ 30$\%$. 

In summary, the discoveries reported here stem from an intuitive hypothesis that the cascading benefits of entanglement-assisted distributed interference suppression could translate to significant capacity gains, particularly when applied to non-linear interference structures like those identified in the literature \cite{clauser1969proposed,ambainis2012quantum,ambainis2012advantage,ambainis2013provable,Briet_and_Vidick}. The magnitude of the advantage, its grounding in Shannon capacity (rather than more fragile notions such as zero-error capacity), and its resilience to substantial quantum noise collectively paint an optimistic picture. Together, these factors open promising avenues to search for concrete practical applications, even in the current NISQ era.

\section{Problem Formulation} \label{sec:problem}
\subsection{Basic definitions}
$\mathbb{N}$ denotes the set of positive integers. For $n\in\mathbb{N}$, $[n]$ denotes the set $\{1,2,\ldots, n\}$, and $x^n$ denotes the tuple $(x_1,x_2,\ldots, x_n)$. For integers $n_1, n_2$, let $[n_1:n_2]$ denote the set $\{n_1,\dots,n_2\}$ if $n_1\leq n_2$ and the empty set otherwise.
Let $\mathbb{R}$ and $\mathbb{C}$ denote the sets of real and complex numbers, respectively, and let $\mathbb{R}_+$ denote the set of nonnegative real numbers. 
We use $H(X)$ and $H(X\mid Y)$ to denote the entropy and the conditional entropy, respectively, and we use $I(X;Y)$ and $I(X;Y\mid Z)$ to denote the mutual information and conditional mutual information, respectively.
$\Pr(E)$ denotes the probability of an event $E$, and $\mathbb{E}[X]$ denotes the expectation of a random variable $X$. We write $a\oplus_m b \triangleq a+b \pmod m$, and omit the subscript $m$ when $m=2$. We write $g(n)=o(f(n))$ if $\lim_{n\to \infty}g(n)/f(n) = 0$. 
We write $g(n)=O(f(n))$ if there exists some constant $c>0$ such that $|g(n)|\leq c|f(n)|$ for all sufficiently large $n$. We write $h(n) = \Omega(f(n))$ if there exists some constant $c>0$ such that $|h(n)|\geq c|f(n)|$ for all sufficiently large $n$.
We write $H_b(p)$ for the binary entropy function, i.e., $H_b(p)=-p\log_2(p) - (1-p)\log_2 (1-p)$, with the convention that $0\log_2 (0) = 0$. 

Let $\mathcal{H}_Q$ denote the Hilbert space associated with a quantum system $Q$. For a composite quantum system $Q_1Q_2\cdots Q_K$, the associated Hilbert space is given by the tensor product $\mathcal{H}_{Q_1}\otimes \mathcal{H}_{Q_2} \otimes \cdots \otimes \mathcal{H}_{Q_K}$. The spaces of linear operators and density operators on $\mathcal{H}_Q$ are denoted by $\mathcal{L}(\mathcal{H}_Q)$ and $\mathcal{D}(\mathcal{H}_Q)$, respectively. 
We write $\Tr[\cdot]$ for the trace of an operator. 
A quantum instrument from input system $Q$ to output system $Q'$, with classical outcome $X\in \mathcal{X}$, is a collection $\mathcal{E} = \{\mathcal{E}_x\}_{x\in \mathcal{X}}$ of completely positive, trace non-increasing maps $\mathcal{E}_x\colon \mathcal{L}(\mathcal{H}_Q) \to \mathcal{L}(\mathcal{H}_{Q'})$ such that $\sum_{x\in \mathcal{X}}\mathcal{E}_x$ is trace-preserving. 
A positive operator-valued measurement (POVM)  on quantum system $Q$, with classical outcome $X\in \mathcal{X}$ is a collection of positive semi-definite operators $\mathcal{D} = \{D_{x}\}_{x\in \mathcal{X}}$ on $\mathcal{H}_Q$ such that $\sum_{x\in \mathcal{X}} D_x = I$, where $I$ is the identity operator on $\mathcal{H}_Q$.

\subsection{$K$-user state-dependent multiple access channel}
A (discrete memoryless) $K$-user state-dependent multiple access channel $(\mN, \mP_{\!S_1S_2\cdots S_K})$ is defined by the channel conditional distribution $\mN(y\mid x_1,x_2,\ldots, x_K, s_1,s_2,\ldots,s_K)$, and the state distribution $\mP_{\!S_1S_2\cdots S_K}(s_1,s_2,\ldots, s_K)$, for all $y \in \mathcal{Y}$, $(x_1,x_2,\ldots, x_K)\in \mathcal{X}_1\times \mathcal{X}_2 \times \cdots \times \mathcal{X}_K$, and $(s_1,s_2,\ldots, s_K)\in \mathcal{S}_1\times \mathcal{S}_2 \times \cdots \times \mathcal{S}_K$, where $\mathcal{Y}$, $\mathcal{X}_k$, and $\mathcal{S}_k$ are finite alphabets for each $k\in [K]$. 
For each channel use, the channel states $(S_1,S_2,\ldots, S_K)$ are generated according to $\mP_{\!S_1S_2\cdots S_K}$, such that $\Pr(S_k=s_k,\forall k\in [K]) = \mP_{\!S_1S_2\cdots S_K}(s_1,s_2,\ldots, s_K)$. For each $k\in [K]$, the $k^{th}$ transmitter (user) observes $S_k$, and chooses an input $X_k \in \mathcal{X}_k$. The output $Y$ of the channel is a random variable seen by the receiver, such that $\Pr(Y=y\mid X_k =x_k, S_k=s_k, \forall k\in [K]) = \mN(y\mid x_1,x_2,\ldots, x_K,s_1,s_2,\ldots, s_K)$.

For $n\in \mathbb{N}$ uses of the channel, we use $S_{k,i}, X_{k,i}$, and $Y_i$ to denote the state $S_k$, input $X_k$, and the output $Y$ associated with the $i^{th}$ use of the channel. The states across different channel uses are independent, i.e., $\Pr(S_k^{n}=s_k^n,\forall k\in [K])=\prod_{i=1}^n \Pr(S_{k,i}=s_{k,i},\forall k)$. The conditional probability of the outputs $\Pr(Y^n=y^n\mid X_k^n=x_k^n, S_k^n=s_k^n,\forall k\in [K]) = \prod_{i=1}^n\Pr(Y_i=y_i\mid X_{k,i}=x_{k,i}, S_{k,i}=s_{k,i},\forall k\in [K])$.

In the following we formalize two classes of coding schemes, namely, the classical coding schemes, and the entanglement-assisted coding schemes, both using causal CSIT.
We then define the achievable rates and the capacities associated with the two classes of coding schemes.

\subsection{Classical coding schemes with causal CSIT}
In a classical coding scheme, the $K$ transmitters and the receiver are allowed access to a shared random variable $J$.
For $M_1,M_2,\ldots, M_K, n \in \mathbb{N}$, a classical $(M_1,M_2,\ldots, M_K,n)$ coding scheme with causal CSIT is specified by a tuple 
\begin{equation*}
	\Big(\mP_{\!J}, \big((\phi_{1}^{(i)}, \phi_{2}^{(i)},\ldots, \phi_{K}^{(i)} )\big)_{i\in [n]},  \psi\Big).
\end{equation*}
Here, $\mP_{\!J}$ denotes the distribution of the shared random variable $J$, whose support is $\mathcal{J}$.
For each $k\in [K]$, Transmitter $k$ sends a message of size $M_k$ by using the channel $n$ times. 
Specifically, for each $k\in [K]$, an independent message $W_k$, uniformly distributed in $[M_k]$, originates at Transmitter $k$. The messages are generated independently of the shared random variable $J$.
For $i=1,2,\ldots, n$, and for each $k\in [K]$, $\phi_{k}^{(i)} \colon [M_k] \times (\mathcal{S}_k)^i \times \mathcal{J} \to \mathcal{X}_k$ specifies the encoder at Transmitter $k$ for the $i^{th}$ channel use, such that $X_{k,i} = \phi_{k}^{(i)}(W_k, S_k^i, J)$. At the receiver, $\psi\colon \mathcal{Y}^n\times \mathcal{J} \to [M_1]\times [M_2] \times \cdots \times [M_K]$ specifies the decoder, with decoded messages $(\widehat{W}_1,\widehat{W}_2,\ldots, \widehat{W}_K) = \psi(Y^n,J)$.

Let us define $\mathfrak{F}^{\C}(M_1,M_2,\ldots, M_K,n)$ as the set of all classical $(M_1,M_2,\ldots, M_K,n)$ coding schemes with causal CSIT.

\subsection{Transmitter-side entanglement-assisted coding schemes with causal CSIT} \label{sec:TxEAcoding}
With entanglement-assistance and causal CSIT available to the transmitters, an $(M_1,M_2,\ldots, M_K,n)$ coding scheme is specified by a tuple 
\begin{equation*}
	\Big(\rho,\big((\mathcal{E}^{(1,i)}, \mathcal{E}^{(2,i)},\ldots, \mathcal{E}^{(K,i)})\big)_{i\in [n]}, \psi \Big).
\end{equation*}
Such a scheme allows the $K$ transmitters to share in advance a $K$-partite (entangled) quantum system $Q_1Q_2\cdots Q_K$ in the initial state $\rho\in \mathcal{D}(\mathcal{H}_{Q_1}\otimes \mathcal{H}_{Q_2}\otimes \cdots \otimes \mathcal{H}_{Q_K})$ that is independent of the messages $(W_1,W_2,\ldots, W_K)$. For each $k\in [K]$, $Q_k$ denotes the quantum system with Transmitter $k$.
The scheme works as follows. For $i=1,2,\ldots,n$, at the $i^{th}$ channel use, and for each $k\in [K]$, conditioned on $(W_k=w_k, S_k^{i}=s_k^{i}, X_k^{i-1}=x_k^{i-1})$, Transmitter $k$ applies the encoding quantum instrument $\mathcal{E}^{(k,i)} = \Big\{\mathcal{E}^{(k,i)}_{x_{k,i}\mid (w_k,s_k^{i},x_k^{i-1})}\Big\}_{x_{k,i}\in \mathcal{X}_k}$, and inputs the classical outcome $X_{k,i}$ into the channel.
After $n$ uses of the channel, the receiver collects the channel outputs $Y^n=(Y_1,Y_2,\ldots, Y_n)$ and applies the decoder $\psi\colon \mathcal{Y}^n \to [M_1]\times [M_2] \times \cdots \times [M_K]$ to produce the decoded messages $(\widehat{W}_1,\widehat{W}_2,\ldots, \widehat{W}_K) = \psi(Y^n)$.
 
Let us define $\mathfrak{F}^{\EA}(M_1,M_2,\ldots, M_K,n)$ as the set of all transmitter-side entanglement-assisted $(M_1,M_2,\ldots, M_K,n)$ coding schemes with causal CSIT. 

\begin{remark}
	Although not considered in this work, an \emph{all-party} entanglement-assisted scheme with causal CSIT may include an additional entangled quantum system accessible to the receiver (decoder). Such schemes are studied for classical point-to-point channels with state in \cite{Yao_Jafar_QEACWS}. 
\end{remark}

Henceforth, we omit the phrases ``with causal CSIT," and ``transmitter-side," as they are the default assumptions throughout the paper.

\subsection{Probability of error, achievable rate and capacity} \label{sec:prob_error_and_capacity}
Given a coding scheme $\Gamma$, the probability of error is defined as 
$P_e(\Gamma) = \Pr\big( (\widehat{W}_1,\ldots,\widehat{W}_K)\neq (W_1,\ldots, W_K) \big)$. The explicit steps to find this value are derived in Appendix \ref{sec:prob_error}.
We remark that, classical shared randomness between the transmitters and the receiver  cannot help to improve the optimal probability of (decoding) error for classical coding schemes (see Remark \ref{rem:randomness} in Appendix \ref{sec:prob_error}).

A rate tuple $(R_1,\ldots, R_K) \in \mathbb{R}_+^K$ is classically  achievable if there exists a sequence of coding schemes $\Gamma_n \in \mathfrak{F}^{\C}(M_{1,n}, \ldots, M_{K,n},n)$ such that $\lim_{n\to \infty} P_e(\Gamma_n) = 0$ and $\lim_{n\to \infty}\log_2(M_{k,n})/n \geq R_k, \forall k\in [K]$.
Similarly, $(R_1,\ldots, R_K) \in \mathbb{R}_+^K$ is achievable with entanglement assistance (EA) if the same conditions hold for a sequence of coding schemes $\Gamma_n \in \mathfrak{F}^{\EA}(M_{1,n}, \ldots, M_{K,n},n)$.

The classical (sum-rate) capacity $C^{\C}$ is defined as $$\sup\{R_1+\cdots +R_K\colon (R_1,\ldots, R_K)~ \mbox{is classically achievable}\}.$$ The entanglement-assisted capacity is defined as $$\sup\{R_1+\cdots +R_K\colon (R_1,\ldots, R_K)~ \mbox{is achievable with EA}\}.$$

\section{Results}
\subsection{Additive MAC with state-dependent interference}
The main results of this paper are developed for the \emph{additive MAC with state-dependent interference} channel model, defined as follows. 
\begin{definition}[Additive MAC with state-dependent interference] \label{def:additive_MAC}
	Let $m\geq 2$ be an integer, and $\mathcal{X}_k=\mathcal{Y} = [0:m-1],\forall k\in [K]$. The channel states $(S_1,\ldots, S_K)$ are generated according to $\mP_{\!S_1\cdots S_K}$.  $Z\in[0:m-1]$ is a random variable  and its distribution is determined by the channel states $(S_1,\ldots, S_K)$, i.e., $\Pr(Z=z\mid S^K=s^K, X^K=x^K) = \Pr(Z=z\mid S^K=s^K)$. The additive MAC with state-dependent interference is then defined by
	\begin{align}
		Y = X_1\oplus_m X_2 \oplus_m \cdots \oplus_m  X_K \oplus_m Z.
	\end{align}
\end{definition}
\begin{remark}
	For  $K=1$ in Definition \ref{def:additive_MAC}, the classical channel capacity is found by Erez and Zamir in \cite{Erez_and_Zamir}. A similar form of the $K$-user MAC (which also includes an additional state output at the receiver) is investigated in \cite[Sec. IV.B]{Yao_Jafar_NS_DoF} when  non-signaling assistance is allowed. In general, however, the capacity of a MAC with causal CSIT is still open \cite{lapidothDoubleState}.
\end{remark}

The goal of this work is to identify the capacity gain from entanglement assistance. To this end, the following theorem characterizes the classical capacity for any additive MAC with state-dependent interference.
\begin{theorem}[Classical capacity] \label{thm:additive_MACs_classical}
	For any additive MAC with state-dependent interference, as defined in Definition \ref{def:additive_MAC},
	\begin{align}
		C^{\C} &= \log_2(m) - H_{\min},\\
		H_{\min} &\triangleq \min_{t_1,\ldots, t_K} H(t_1(S_1)\oplus_m \cdots \oplus_m t_K(S_K)\oplus_m Z),\label{eq:min_obj}
	\end{align}
	and the minimization is over the functions $t_k\colon \mathcal{S}_k \to \mathcal{X}_k, \forall k$.
\end{theorem}
\noindent The proof of Theorem \ref{thm:additive_MACs_classical} is presented in Appendix \ref{proof:additive_MACs_classical}. Theorem \ref{thm:additive_MACs_classical} can be viewed as an extension of \cite[Thm. 1]{Erez_and_Zamir} to the MAC setting. The converse follows from a standard application of Fano's inequality. The achievability can be interpreted as the following \emph{channel conversion} strategy which tries to suppress the interference over each channel use. Specifically, for each $k\in [K]$, Transmitter $k$ sends $X_k = X_k'\oplus_m t_k(S_k)$ for some $X_k'\in \mathcal{X}_k$. The receiver sees $Y = X_1'\oplus_m \cdots \oplus_m X_K'\oplus (t_1(S_1) \oplus_m \cdots \oplus_m t_K(S_K) \oplus Z)$. The term $(t_1(S_1) \oplus_m \cdots \oplus_m t_K(S_K) \oplus_m Z)$ is then viewed as the \emph{suppressed} interference, which has entropy $H_{\min}$ when minimized over all deterministic functions $t_k \colon \mathcal{S}_k \to \mathcal{X}_k, \forall k$. By coding over the converted channel $(X_1',\ldots, X_K') \to Y$, the sum-rate $\log_2(m) - H_{\min}$ is achievable. Note that the channel conversion idea only uses CSIT of the current channel, and therefore it only needs causal CSIT. We also observe that when $H_{\min} = H(Z)$, the classical capacity with causal CSIT coincides with the classical capacity without CSIT.

We next extend the channel conversion idea to incorporate transmitter-side entanglement assistance. This is formalized in Theorem \ref{thm:EAIS}.
\begin{theorem}[Entanglement-assisted distributed interference suppression] \label{thm:EAIS}
	Let $\rho\in \mathcal{D}(\mathcal{H}_{Q_1}\otimes \cdots \otimes \mathcal{H}_{Q_K})$ be any (possibly) entangled state of any $K$-partite composite quantum system $Q_1\cdots Q_K$, and let $\mathcal{F}^{(k)}_{s_k} = \big\{F^{(k)}_{b_k \mid s_k}\big\}_{b_k\in [0:m-1]}$ be any POVM on $Q_k$ which has classical output defined in $[0:m-1]$, $\forall k\in [K], s_k\in \mathcal{S}_k$. Let $B_k$ denote the measurement outcome when $\mathcal{F}_{S_k}^{(k)}$ is applied to $Q_k,\forall k$. Then the following rate $R^{\EA}$ is achievable with transmitter-side entanglement-assisted schemes.
	\begin{align}
		R^{\EA} &= \log_2(m) - H(B_1\oplus_m \cdots \oplus_m B_K \oplus_m Z).
	\end{align}
\end{theorem}
\noindent The proof is presented in Appendix \ref{proof:EAIS}.

\subsection{2-user cases}
As an introductory example to see how entanglement assistance can improve the capacity with causal CSIT, let us consider the cases with $K=2$ users, and present the following example, namely \textit{Channel A1}.
\begin{definition}[Channel A1] \label{def:ChannelA1}
	Channel A1 is an additive MAC with state-dependent interference. It has $K=2$ users and $m=2$. The state alphabets are $\mathcal{S}_1=\mathcal{S}_2 = \{0,1\}$. The states $S_1$ and $S_2$ are drawn independently and uniformly from $\mathcal{S}_1$ and $\mathcal{S}_2$, respectively. The state-dependent interference $Z= S_1  S_2$, i.e., $Z=1$ if and only if $S_1=1$ and $S_2=1$. 
\end{definition}
\noindent We have the following proposition regarding Channel A1.
\begin{proposition}[Channel A1] \label{prop:ChannelA1}
	For Channel A1, the classical capacity is $C^{\C} = 1-H_b(3/4) \approx 0.1887$, and its entanglement-assisted capacity $C^{\EA} \geq 1-H_b(\tfrac{2+\sqrt{2}}{4}) \approx 0.3991$. Thus, the multiplicative capacity gain via entanglement assistance for Channel A1 is at least $0.3991/0.1887 \approx 2.1$.  Fig. \ref{fig:Hbp_complement} illustrates $1-H_b(p)$.
\end{proposition}
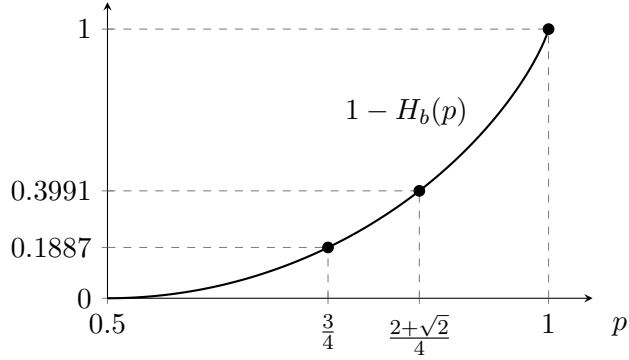
\begin{figure}[htbp]
\center
\begin{tikzpicture}
\begin{axis}[
    width=8cm,
    height=5.5cm,
    axis lines=left,
    xmin=0.5, xmax=1.05,
    ymin=0, ymax=1.1,
    xlabel={$p$},
    ylabel={},
    xlabel style={
        at={(axis description cs:1.02,0.04)},
        anchor=west
    },
    xtick={0.5,0.75,0.8536,1},
    xticklabels={$0.5$,$\frac{3}{4}$, $\frac{2+\sqrt{2}}{4}$, $1$},
    ytick={0,0.1887,0.3991,1},
    yticklabels={$0$, $0.1887$, $0.3991$, $1$},
    samples=200,
    domain=0.5:1,
    clip=false,
]

% Function 1 - H_b(p)
\addplot[
    thick,
    black
]
{1 + x*log2(x) + (1-x)*log2(1-x)};
\node[above left] at (axis cs:0.92,0.60) {$1-H_b(p)$};

% Point at p = 3/4
\addplot[
    only marks,
    mark=*,
    mark size=2pt,
    black
]
coordinates {(0.75,0.1887)};
\addplot[
    dashed,
    gray
]
coordinates {(0.5,0.1887) (3/4,0.1887)};
\addplot[
    dashed,
    gray
]
coordinates {(3/4,0) (3/4,0.1887)};

% Point at p = 1/2 + sqrt(2)/4
\addplot[
    only marks,
    mark=*,
    mark size=2pt,
    black
]
coordinates {({1/2 + sqrt(2)/4},0.3991)};
\addplot[
    dashed,
    gray
]
coordinates {({1/2 + sqrt(2)/4},0) ({1/2 + sqrt(2)/4},0.3991)};
\addplot[
    dashed,
    gray
]
coordinates {(0.5,0.3991) ({1/2 + sqrt(2)/4},0.3991)};

% Point at p = 1
\addplot[
    only marks,
    mark=*,
    mark size=2pt,
    black
]
coordinates {(1,1)};

\addplot[
    dashed,
    gray
]
coordinates {(1,0) (1,1)};
\addplot[
    dashed,
    gray
]
coordinates {(0.5,1) (1,1)};

\end{axis}
\end{tikzpicture}
\caption{Plot of $1-H_b(p)$ for $p\in [0.5,1]$. Note that $H_b(p) = H_b(1-p)$ for $p\in [0,1]$.} 
\label{fig:Hbp_complement}
\end{figure}

\begin{proof}[Proof of Proposition \ref{prop:ChannelA1}]
	First, let us apply Theorem \ref{thm:additive_MACs_classical} to obtain $C^{\C}$. For any fixed functions $t_1\colon \{0,1\}\to \{0,1\}$ and $t_2\colon \{0,1\}\to \{0,1\}$, let $V_{t_1,t_2} \triangleq t_1(S_1)\oplus t_2(S_2)\oplus S_1S_2$ be a binary random variable distributed in $\{0, 1\}$. 
	We note that minimizing $H(V_{t_1,t_2})$ over functions $(t_1,t_2)$ corresponds to  maximizing the probability $\Pr(V_{t_1,t_2}=0)$.
	Let $\Sigma_0 = \{(s_1,s_2)\in \{0,1\}^2\colon t_1(s_1)\oplus t_2(s_2) \oplus s_1s_2 = 0\}$. It can be verified that $\bigoplus_{s_1,s_2} (t_1(s_1)\oplus t_2(s_2)\oplus s_1s_2) = 1$, regardless of $(t_1,t_2)$. It follows that $|\Sigma_0| \in \{1,3\}$. Since $(S_1,S_2)$ are uniformly distributed, $\Pr(V_{t_1,t_2}=0) = \tfrac{1}{4}|\Sigma_0| \in \{\tfrac{1}{4},\tfrac{3}{4}\}$. Since $H_b(\frac{1}{4}) = H_b(\frac{3}{4})$, it follows that $H(V_{t_1,t_2}) = H_b(3/4)$, for all functions $(t_1,t_2)$, and therefore $H_{\min} = H_b(3/4)$. Applying Theorem \ref{thm:additive_MACs_classical}, we have $C^{\C} = 1-H_b(3/4)$.  
	Next, let us show that the rate $1-H_b(\tfrac{2+\sqrt{2}}{4})$ is achievable with transmitter-side entanglement assistance. For each channel use, let $X_k' \in \{0,1\}$ be some variable that can be set by Transmitter $k$ for each $k\in \{1,2\}$. Based on $S_k$, Transmitter $k$ measures its entangled resource and obtains $B_k\in \{0,1\}$. It then sends $X_k = X_k' \oplus B_k$ so that the receiver sees $Y = X_1\oplus X_2 \oplus S_1S_2 = X_1'\oplus X_2' \oplus (B_1\oplus B_2 \oplus S_1S_2)$. It is known that there exists an entanglement-assisted strategy to obtain such $(B_1,B_2)$ that yield $\Pr(B_1\oplus B_2 \oplus S_1S_2 = 0) = \tfrac{2+\sqrt{2}}{4}$. This is known as the CHSH strategy \cite{clauser1969proposed}, in which the entangled resource consumed is a pair of qubits in the Bell state $\ket{\Phi^+} = \frac{1}{\sqrt{2}}(\ket{00}+\ket{11})$. We delegate the details to Appendix \ref{proof:CHSH}. Now, note that the receiver sees $Y = X_1'\oplus X_2' \oplus Z'$, where $Z'\triangleq B_1\oplus B_2 \oplus S_1S_2$ can be viewed as the suppressed interference, which has entropy $H(Z') = H_b(\tfrac{2+\sqrt{2}}{4})$. Let $R_2=0$ and set $X_2'=0$. It then follows from the direct coding theorem of a point-to-point channel that, by choosing $X_1'$ uniformly over $\{0,1\}$, any rate $R_1 < I(X_1';X_1' \oplus Z') = H(X_1' \oplus Z') - H(X_1' \oplus Z' \mid X_1') = 1 - H(Z')$ is achievable for $W_1$. Due to symmetry and by a TDMA scheduling scheme, any rate tuple $(R_1,R_2)$ satisfying $R_1+ R_2<1-H(Z')$ is also achievable. We point out that the achievable rate $1-H_b(\tfrac{2+\sqrt{2}}{4}) \approx 0.3991$ requires only transmitter-side entangled resources, and that once an entangled resource is consumed in one channel use, it does not need  to be carried to the next channel use, i.e., no cross-channel-use entanglement is needed.
\end{proof}

\begin{remark}[MAC without state] \label{rem:MAC_no_state}
	It follows from Theorem \ref{thm:MAC_no_state} provided in Section \ref{sec:MAC_no_state}, that for a $K$-user MAC without state, the maximum multiplicative capacity gain from entanglement assistance (or even with non-signaling assistance) cannot be more than $K$. In fact, the largest multiplicative gain  in Shannon capacity from entanglement assistance in a $2$-user MAC without state  noted in the literature thus far to our knowledge is only about $1.05$ \cite{seshadri2023separation, QMAC_Yun}, i.e., the additive gain in sum-rate is about $5\%$ of the classical capacity. In contrast, we see on Channel A1 that for a $2$-user state-dependent MAC, the multiplicative gain is strictly larger than $2$, i.e., the additive gain in sum-rate is $>100\%$ of the classical capacity. 	
\end{remark}

The following lemma is an elementary result on the asymptotic behavior of the function $1-H_b(p)$ around $p= 0.5$. This will be helpful to understand the rest of the results.
\begin{lemma} \label{lem:approximation}
	$1-H_b(\frac{1}{2}+\epsilon) = \frac{2}{\ln(2)}\epsilon^2 + \delta(\epsilon)$,  and $\lim_{\epsilon \to 0}\frac{\delta(\epsilon)}{\epsilon^2} = 0$. In addition, $1-H_b(\frac{1}{2}+\epsilon) \geq \frac{2}{\ln(2)}\epsilon^2, \forall \epsilon \in [-1/2,1/2]$.
\end{lemma}
\begin{proof}
	Let $g(p)=1-H_b(p)$, so that $1-H_b(\frac{1}{2}+\epsilon) = g(\frac{1}{2}+\epsilon)$. For $p\in (0,1)$, $g(p)= 1+p\log_2(p) +(1-p)\log_2(1-p)$ is analytic. We have for $p\in (0,1)$, $g'(p) = \log_2(p)-\log_2(1-p), g''(p) =\frac{1}{\ln(2)}(\frac{1}{p}+\frac{1}{1-p})$, $g'''(p) = \frac{1}{\ln(2)}(-\frac{1}{p^{2}}+\frac{1}{(1-p)^2})$, $\cdots$, $g^{(k)}(p) = \frac{(k-2)!}{\ln(2)}\big(\frac{(-1)^k}{p^{k-1}}+\frac{1}{(1-p)^{k-1}}\big)$ for $k\geq 2$, where $g^{(k)}(p)$ denotes the $k^{th}$ order derivative of $g$. In particular, $g^{(k)}(1/2) = 0$ for odd $k$, whereas $g^{(k)}(1/2) > 0$ for even $k$.
	 The Taylor expansion of $g(p)$ at $p=1/2$ gives $g(\frac{1}{2}+\epsilon)  = \frac{2}{\ln(2)}\epsilon^2 + \delta(\epsilon)$ where $\lim_{\epsilon \to 0} \frac{\delta(\epsilon)}{\epsilon^2}=0$. Moreover, all the remaining Taylor coefficients are non-negative, and thus   $g(\frac{1}{2}+\epsilon)\geq \frac{2}{\ln(2)}\epsilon^2$ for $\epsilon \in (-1/2,1/2)$. The endpoint cases $\epsilon=\pm 1/2$ follow directly from $1-H_b(1)=1-H_b(0)=1 > \frac{1}{2\ln(2)}$. 
\end{proof}

Next, let us define a particular subset of $2$-user additive MACs with state-dependent interference, namely \emph{Class A}  as follows.
\begin{definition}[Class A] \label{def:Class_A}
	A channel in Class A has $K=2$ users and $m=2$ (inputs and outputs are binary). The states $S_1$ and $S_2$ are independently and uniformly drawn from $\mathcal{S}_1$ and $\mathcal{S}_2$, respectively, where $|\mathcal{S}_1| = |\mathcal{S}_2| = L$ for some $L\in \mathbb{N}$. The state-dependent interference $Z = f(S_1,S_2)$, for a function $f\colon \mathcal{S}_1\times \mathcal{S}_2 \to \{0,1\}$.
\end{definition}
Without loss of generality, consider the channels in Class A for which $\mathcal{S}_1=\mathcal{S}_2 = [L]$. For $L\in \mathbb{N}$, suppose one randomly picks such a channel $\ch_L$ by assigning to $f(s_1,s_2)$ a value  uniformly in $\{0,1\}$, independently for each $(s_1,s_2)\in [L]^2$. Denote by $C^{\C}(\ch_L)$ and $C^{\EA}(\ch_L)$ the classical capacity and the entanglement-assisted capacity associated with $\ch_L$, respectively. We then have the following proposition. 
\begin{proposition}[Class A] \label{prop:ClassA}
	As $L\to \infty$, with probability $1-o(1)$, $C^{\C}(\ch_L) \leq 2L^{-1}+o(L^{-1})$, and $C^{\EA}(\ch_L) \geq \frac{2}{\ln(2)}L^{-1} + o(L^{-1})$. This implies that the multiplicative capacity gain from entanglement assistance for (asymptotically) almost every channel in Class A (as $L\to \infty$) is at least $1/\ln(2)\approx 1.44$.
\end{proposition}
\noindent The proof is presented in Appendix \ref{proof:use_xor_games}. It is based on a result by Ambainis et al. \cite{ambainis2012quantum}.

\subsection{Exponential capacity gain in number of users $K$ with fixed (binary) state alphabet}
Now let us shift our focus to $K\geq 3$ users, and present an example for which we witness an exponential (in $K$) gain in capacity from transmitter-side  entanglement assistance. The examples in this section require only binary alphabet for the states.
\begin{definition}[Channel B1] \label{def:MerminGHZ_channel}
	Channel B1 has $K\geq 3$ users and $m=2$ (binary inputs and outputs). The states are binary and correlated so that $S_K = S_1 \oplus S_2 \oplus \cdots \oplus S_{K-1}\in\{0,1\}$, where $S_1,\ldots, S_{K-1}$ are independent, and uniform over $\{0,1\}$. The state-dependent interference is $Z=(S_1+S_2+\cdots +S_K)/2 \pmod 2$. 
\end{definition}
Note that the integer sum $S_1+S_2+\cdots+S_K$ yields an even number because of the correlation relationship, ensuring that $(S_1+S_2+\cdots +S_K)/2$ is an integer and then the mod 2 operation produces an output in $\{0,1\}$. The particular structure of $Z$ is inspired by the Mermin-GHZ game \cite{merminGHZ}.

The next proposition establishes an exponential gain from quantum entanglement assistance.
\begin{proposition}[Channel B1] \label{prop:ChannelB1}
	For Channel B1, the classical and entanglement-assisted capacities are as follows.
	\begin{align}
	C^{\C} &= 1-H_b(\tfrac{1}{2}+ 2^{-\lceil K/2 \rceil}) = O(2^{-K}), \label{eq:CC_ChannelB1} \\
	C^{\EA} &= 1.
	\end{align}
\end{proposition}
\noindent The proof of Proposition \ref{prop:ChannelB1} is presented in Section \ref{proof:MerminGHZ}.  
We plot the values of $C^{\EA}/C^{\C}$ for $K\in [3:8]$ in Fig. \ref{fig:gain}.

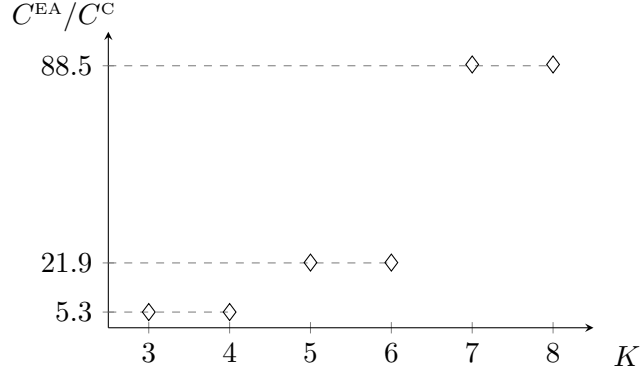
\begin{figure}[htbp]
\center
\begin{tikzpicture}
	\begin{axis}[
    width=8cm,
    height=5.5cm,
    axis lines=left,
    xmin=2.5, xmax=8.5,
    ymin=0, ymax=100,
    xlabel={$K$},
    ylabel={$C^{\EA}/C^{\C}$},
    xlabel style={
        at={(axis description cs:1.02,0.04)},
        anchor=west
    },
    ylabel style={
    	rotate =270,
    	at={(axis description cs:0.1,1.05)}
    },
    xtick={3,4,5,6,7,8},
    xticklabels={$3$,$4$,$5$,$6$,$7$,$8$},
    ytick={5.3,21.9,88.5},
    yticklabels={$5.3$,$21.9$,$88.5$},
    samples=200,
    domain=0.5:1,
    clip=false,
]

% C_{EA}/C_{C}
\addplot[
    only marks,
    mark=diamond,
    mark size=3.2pt,
    samples at={3,4,5,6,7,8},
]
{(1+(1/2+2^(-ceil(x/2)))*log2(1/2+2^(-ceil(x/2)))+(1/2-2^(-ceil(x/2)))*log2(1/2-2^(-ceil(x/2))))^(-1)};

% lines
\addplot[
    dashed,
    gray
]
coordinates {(2.5,5.3) (4,5.3)};
\addplot[
    dashed,
    gray
]
coordinates {(2.5,21.9) (6,21.9)};
\addplot[
    dashed,
    gray
]
coordinates {(2.5,88.5) (8,88.5)};

\end{axis}
\end{tikzpicture}
\caption{The ratio of entanglement-assisted  capacity to classical capacity, $C^{\EA}/C^{\C}$ is illustrated for Channel B1, versus the number of users $K$.}
\label{fig:gain}
\end{figure}

While the correlation among channel states that is assumed in Channel B1 does simplify our analysis, allowing us to find the exact capacities in Proposition \ref{prop:ChannelB1}, it is not an essential requirement for an exponential advantage. Our next example  will show that the exponential gain also exists for a channel where the channel states $(S_1,\ldots, S_K)$ are independent. See Definition \ref{def:ChannelB2}, which defines such a channel Channel B2, and Proposition \ref{prop:ChannelB2}.
\begin{definition}[Channel B2] \label{def:ChannelB2}
	Channel B2 has $K\geq 3$ users and $m=2$ (binary inputs and outputs). $S_1,S_2,\ldots, S_K$ are independent, and each is uniformly distributed over $\{0,1\}$. The state-dependent interference is $Z=\lceil (S_1+S_2+\cdots +S_K)/2\rceil \pmod 2$.
\end{definition}
\begin{proposition}[Channel B2] \label{prop:ChannelB2}
	For Channel B2, 
	\begin{align}
		C^{\C} \leq 1-H_b(\tfrac{1}{2}+ 2^{-\lceil (K+1)/2 \rceil}) = O(2^{-K}),
	\end{align}
	whereas 
	\begin{align}
		C^{\EA} \geq 1-H_b(3/4)\approx 0.1887.
	\end{align}
\end{proposition} 
\noindent The proof of Proposition \ref{prop:ChannelB2} is presented in Appendix \ref{proof:Mermin_indep}.

\begin{remark}[Robustness] \label{rem:robust}
	The entanglement-assisted schemes to achieve $C^{\EA}$ for Channel B1 (and the lower bound of $C^{\EA}$ for Channel B2) only consume $K$ entangled qubits (one per transmitter) over each channel use, and no cross-channel-use entanglement is needed. Moreover, we show in Appendix \ref{proof:robust} that the exponential gain over $C^{\C}$ persists even if each qubit independently has $\approx 30\%$ probability of being completely depolarized. 
\end{remark}

Next, let us define a particular subset of additive MACs with state-dependent interference, namely \textit{Class $B$} as follows. 
\begin{definition}[Class B]
	A channel in Class B has $K$ users and $m=2$ (binary inputs and outputs). The states $S_1,S_2,\ldots, S_K$ are drawn from $\{0,1\}$ independently uniform. The state-dependent interference is $Z = g(S_1,\ldots,S_K)$, where $g\colon \{0,1\}^K \to \{0,1\}$ is a `symmetric' function, such that for any $(s_1,\ldots, s_K)\in \{0,1\}^K$ and any permutation $\pi\colon [K] \to [K]$, $g(s_1,\ldots, s_K) = g(s_{\pi(1)},\ldots, s_{\pi(K)})$. Equivalently, such a channel is specified by $K$ and a tuple $(G_0,G_1,\ldots, G_K)\in \{0,1\}^{K+1}$ such that $g(s_1,\ldots, s_K) = G_{s_1+\cdots +s_K}$.
\end{definition}
Suppose one randomly picks a $K$-user channel $\ch_K$ in Class B, by assigning $G_i$ to a value uniformly in $\{0,1\}$, independently for each $i\in [0:K]$. Denote by $C^{\C}(\ch_K)$ and  $C^{\EA}(\ch_K)$ the classical capacity and the entanglement-assisted capacity for $\ch_K$, respectively. We then have the following proposition.
\begin{proposition}[Class B] \label{prop:ClassB}
	As $K\to \infty$, with probability $1-o(1)$, $C^{\C}(\ch_K) = O(K^{-1/2})$, whereas $C^{\EA} = \Omega\big(\log(K)K^{-1/2}\big)$. This implies that the multiplicative gain from entanglement assistance for (asymptotically) almost every channel in Class B (as $K\to \infty$) is at least on the order of $\log(K)$.
\end{proposition}
\noindent The proof is presented in Appendix \ref{proof:use_xor_games}. It is based on the results by Ambainis et al. \cite{ambainis2012advantage, ambainis2013provable}. 

\subsection{Unbounded capacity gain for large state alphabet with fixed ($K=3$) number of users}
The capacity gain factors we saw in Class B are bounded if $K$ is fixed. One may wonder if the gain factor can be actually unbounded with finite number of users. We answer this question in the affirmative in the following proposition.
\begin{proposition}[Unbounded] \label{prop:ClassC}
	For $r\in \mathbb{N}$, there exists a sequence of $(K=3, m=2)$ additive MACs with state-dependent interference, $\{\ch_r\}_r$, where $\ch_r$ has $|\mathcal{S}_1|=|\mathcal{S}_2|=|\mathcal{S}_3| = 4^{r}$, such that $C^{\EA}(\ch_r)/C^{\C}(\ch_r) \geq \Omega\big(2^{r} r^{-5}\big)$. Let us use Class C to refer to the channels in this sequence.
\end{proposition}
\noindent The proof is presented in Appendix \ref{proof:use_xor_games}. It is based on a result by Bri\"{e}t and Vidick \cite{Briet_and_Vidick}. We point out that the channels used to show the existence of the gain may have dependent channel states $(S_1,S_2,S_3)$. 

The table below summarizes the multiplicative gains $C^{\EA}/C^{\C}$ in the aforementioned channels (classes). Note that in all cases considered, the channel input and output alphabet are binary.

\begin{table}[htbp]
\caption{Summary of the capacity gains via entanglement assistance}
\center
\begin{tabular}{|l|l|l|}
\hline
Channel/Class & Description (inputs and outputs are always binary)& $C^{\EA}/C^{\C}$ \\ \hline
Channel A1 & $2$ users, $(S_1,S_2)\sim {\rm Unif}(\{0,1\}^2)$, $Z=S_1S_2$ & $\geq 2.1$ \\ \hline
Class A (generic) & $2$ users, $L$-ary indep. states, $Z=f(S_1,S_2)$ & $\geq 1.44$ w.h.p. \\ \hline
Channel B1 & $K$ users, binary dependent states & $\Omega(2^K)$\\ \hline
Channel B2 & $K$ users, binary indep. states & $\Omega(2^K)$\\ \hline
Class B (generic) & $K$ users, binary indep. states, symmetric interference & $\Omega(\log(K))$ w.h.p.\\ \hline
Class C & $3$ users, $4^r$-ary dependent states & $\Omega(2^{r}r^{-5})$ \\ \hline
\end{tabular}
\end{table}

\subsection{MAC without state} \label{sec:MAC_no_state} 
The purpose of this subsection is to provide an \emph{impossibility} result for the capacity gain from quantum entanglement assistance in MACs \emph{without state}. Consider $K$-user MACs \emph{without state}, i.e., specified by their channel conditional distribution $\mN(y\mid x_1,x_2,\ldots, x_K)$. Fawzi and Ferm\'{e} \cite{fawzi2024MAC} studied the capacity for $2$-user MACs with non-signaling assistance (which strictly includes quantum entanglement  assistance)  for all parties ($2$ transmitters and $1$ receiver). Let us extend the setting to the $K$-user MACs, and denote the $(K+1)$-partite non-signaling assisted (sum-rate) capacity as $C^{\NS}$. Clearly, $C^{\NS} \geq C^{\EA}$. We have the following theorem.

\begin{theorem} \label{thm:MAC_no_state}
	For any $K$-user MAC without state, $\mN(y\mid x_1,x_2,\ldots, x_K)$, we have $C^{\NS} \leq KC^{\C}$. Therefore, for any $K$-user MAC without state, $C^{\EA} \leq C^{\NS}\leq KC^{\C}$.
\end{theorem}
\begin{proof}[Proof of Theorem \ref{thm:MAC_no_state}]
We are given the $K$-user MAC, i.e., $\mN(y\mid x_1,x_2,\ldots, x_K), \forall x_k \in \mathcal{X}_k, k\in [K], y\in \mathcal{Y}$. For  any rate tuple $(R_1,R_2,\ldots, R_K)$ achievable with non-signaling assistance, by considering $(X_2,\ldots, X_K)\triangleq X_{1^c}$ as the input for one (big) transmitter,  \cite[Prop. 41]{fawzi2024MAC} implies that
\begin{align}
	R_1 &\leq \max_{\mP_{\!X_1X_{1^c}}} I(X_1;Y\mid X_{1^c}) \label{eq:max_joint}  \\
	&=\max_{\mP_{\!X_1\mid X_{1^c}}} \max_{\mP_{\!X_{1^c}}} \sum_{x_{1^c}\in \mathcal{X}_{1^c}}\mP_{\!X_{1^c}}(x_{1^c}) I(X_1;Y\mid X_{1^c}=x_{1^c}) \label{eq:max_conditional_1} \\
	&=  \max_{\mP_{\!X_1\mid X_{1^c}}}  \max_{x_{1^c} \in \mathcal{X}_{1^c}}  I(X_1;Y\mid X_{1^c}=x_{1^c}) \label{eq:max_conditional_2}\\
	&= \max_{x_{1^c}} \max_{\mP_{\!X_1}}  I(X_1;Y\mid X_{1^c}=x_{1^c}) \label{eq:two_user_C1} \\
	&\triangleq R_{1}^{\max}
\end{align}
In \eqref{eq:max_joint}, $\mP_{\!X_1X_{1^c}}$ is any joint distribution of $(X_1,X_{1^c}) = (X_1,X_2,\ldots, X_K)$. In \eqref{eq:max_conditional_1} and \eqref{eq:max_conditional_2}, $\mathcal{X}_{1^c} \triangleq \mathcal{X}_2 \times \cdots \times \mathcal{X}_K$, $\mP_{\!X_1\mid X_{1^c}}$ is any conditional distribution of $X_1$ given $X_{1^c}$, and $\mP_{\!X_{1^c}}$ is any distribution of $X_{1^c}$. In \eqref{eq:two_user_C1}, $\mP_{\!X_1}$ is any distribution of $X_1$.

Note from \eqref{eq:two_user_C1} that any rate tuple $(R_1,0,\ldots, 0)$ where $R_1< R_1^{\max}$ is achievable \emph{classically} by setting $X_{1^c}=x_{1^c}$ that maximizes \eqref{eq:two_user_C1} and coding over the remaining point-to-point channel $X_1\to Y$ for $W_1$.
Now, for $k\in [K]$, define $C_k^{\C} \triangleq \sup \{R_k\colon (R_1,\ldots, R_K) \mbox{ is achievable classically}\}$. 
We have that $R_1 \leq R_1^{\max} \leq C_1^{\C}$, if $(R_1,\ldots, R_K)$ is achievable with  non-signaling assistance. In other words, non-signaling assistance cannot improve the highest rate that can be achieved for User $1$. By symmetry, for any $(R_1,\ldots, R_K)$ that is achievable with non-signaling assistance, $R_1+R_2+\cdots+R_K \leq C_1^{\C} + C_2^{\C} + \cdots + C_K^{\C} \leq  K\max_{k\in[K]} \{C_k^{\C}\}$, and thus $C^{\NS} \leq K\max_{k\in[K]} \{C_k^{\C}\}$. Since the classical (sum-rate) capacity $C^{\C} \geq \max_{k\in[K]} \{C_k^{\C}\}$, we conclude that $C^{\NS} \leq KC^{\C}$ for any $K$-user MAC (without state). 
\end{proof}

\section{Conclusion}
Exponential, unbounded and robust gains in capacity are established for various $K$-user multiple access channels with causal CSIT, due to quantum entanglement assistance provided to the transmitters. 
 This points to a rich landscape of  open questions for future work regarding how  the gains extend beyond the specific examples considered here. For example, it is not known if similar gains exist when the CSIT  is \emph{strictly} causal \cite{lapidothDoubleState}, non-causal \cite{NoncausalCSIT}, or mixed as in \cite{MAC_mixed_states}. While it is evident that the gains can be further improved in some cases by providing entanglement assistance also to the receiver (because such gains have already been established in \cite{Yao_Jafar_QEACWS} for the point-to-point channel with causal CSIT, which is a special case of the MAC with causal CSIT), it is not yet known how significant this additional gain could be. More broadly, it is not known how the gains generalize beyond the MAC  to other state-dependent communication networks with quantum entanglement assistance.

\appendix
\section{Probability of error of a coding scheme} \label{sec:prob_error}
Recall that we are given the $K$-user state-dependent MAC $(\mN, \mP_{\!S_1S_2\cdots S_K})$. 
For a coding scheme $\Gamma\in \mathfrak{F}^{\C}(M_1,\ldots, M_K,n) \cup \mathfrak{F}^{\EA}(M_1,\ldots, M_K,n)$, the probability of error is $P_e(\Gamma) = \Pr\big( (\widehat{W}_1,\ldots,\widehat{W}_K)\neq (W_1,\ldots, W_K) \big)$. In the following we derive an expression for $1-P_e(\Gamma)$.

Let us use bold letters to denote the corresponding $K$-tuples. For example, let ${\bf W}$ represent the $K$-tuple $(W_1, W_2,\ldots, W_K)$. Then,
\begin{align}
	&1-P_e(\Gamma) = \Pr \big( \widehat{\bf W} = {\bf W}\big) \\
	&= \sum_{{\bf w}, {\bf s}^n} \Pr({\bf W} = {\bf w}) \Pr({\bf S}^n={\bf s}^n) \Pr\big( \widehat{\bf W} = {\bf w} \mid {\bf W}={\bf w}, {\bf S}^n={\bf s}^n\big) \label{eq:prob_error_first}
\end{align}
since ${\bf W} = (W_1,W_2,\ldots, W_K)$ is independent of ${\bf S}^n = (S_1^n,S_2^n,\ldots, S_K^n)$. Note that $\Pr({\bf W}={\bf w}) = (M_1M_2\cdots M_K)^{-1}$ since the messages are generated independently uniform. Meanwhile,  $\Pr({\bf S}^n = {\bf s}^n) = \prod_{i=1}^n \mP_{\!S_1S_2\cdots S_K}(s_{1,i},s_{2,i},\ldots, s_{K,i})$. Clearly, $\Pr({\bf W} = {\bf w}) \Pr({\bf S}^n={\bf s}^n)$ does not depend on $\Gamma$, and in the following we derive $\Pr\big( \widehat{\bf W} = {\bf w} \mid {\bf W}={\bf w}, {\bf S}^n={\bf s}^n\big)$.

We simply write $a$ for the event $A=a$. Consider two cases.

\noindent{\bf Case I:} If  $\Gamma = \big(\mP_{\!J}, ((\phi_{1}^{(i)}, \phi_{2}^{(i)},\ldots, \phi_{K}^{(i)} ))_{i\in [n]},  \psi\big)  \in \mathfrak{F}^{\C}(M_1,M_2,\ldots, M_K,n)$ is a classical coding scheme with shared randomness $J$ (which is independent of $({\bf W}, {\bf S}^n)$, then 
\begin{align}
	&\Pr\big( \widehat{\bf W} = {\bf w} \mid {\bf W}={\bf w}, {\bf S}^n={\bf s}^n\big)\notag\\
	&=\sum_{j\in \mathcal{J}}\mP_{\!J}(j)  \Pr\big( \widehat{\bf W} = {\bf w} \mid {\bf w}, {\bf s}^n, j\big)\\
	&=\sum_{j\in \mathcal{J}}\mP_{\!J}(j) \sum_{y^n\colon \psi(y^n,j) = {\bf w} } \prod_{i=1}^n  \mN\Big(y_i \mid \phi_1^{(i)}(w_1,s_1^i,j),\ldots, \phi_K^{(i)}(w_K,s_K^i,j), s_{1,i},\ldots, s_{K,i}\Big) 
	\label{eq:prob_error_classical}
\end{align}
To see the last step, note that for each $k\in [K]$, $i\in [n]$, and $j\in \mathcal{J}$, $\phi_k^{(i)}(w_k,s_k^i,j)$ is the input at Transmitter $k$ for the $i^{th}$ channel use, and the sum is taken over the subset of $y^n$ for which the decoder $\psi(y^n,j)$ outputs the correct messages ${\bf w}$.

\begin{remark}[Classical randomness is not helpful] \label{rem:randomness}
	Note that for a classical scheme $\Gamma$, $P_e(\Gamma)$ can be expressed as $\sum_{j\in \mathcal{J}}\mP_{\!J}(j)p(j)$ for some function $p\colon \mathcal{J}\to [0,1]$, which must be lower bounded by $p(j^*)$ for some $j^* \in \arg\min_j p(j)$. This means that the probability of error of a deterministic coding scheme by conditioning on $J=j^*$ in $\Gamma$ cannot be worse (bigger) than $P_e(\Gamma)$. Therefore, the minimum of $P_e(\Gamma)$ over $\mathfrak{F}^{\C}(M_1,M_2,\ldots, M_K,n)$ is achieved by some deterministic coding scheme, i.e., without any randomness.
\end{remark}

\noindent {\bf Case II:} If $\Gamma = \big(\rho,((\mathcal{E}^{(1,i)}, \mathcal{E}^{(2,i)},\ldots, \mathcal{E}^{(K,i)}))_{i\in [n]}, \psi \big) \in \mathfrak{F}^{\EA}(M_1,M_2,\ldots, M_K,n)$ is an entanglement-assisted coding scheme, then
\begin{align}
	&\Pr\big(\widehat{\bf W}={\bf w} \mid {\bf W}={\bf w}, {\bf S}^n={\bf s}^n) \notag \\
	&=\sum_{{\bf x}^n, y^n}\Pr\big(\widehat{\bf W}={\bf w}, {\bf x}^n, y^n \mid {\bf w},  {\bf s}^n) \\
	&=\sum_{{\bf x}^n, y^n}\Pr\big( {\bf x}^n \mid {\bf w},  {\bf s}^n)  \Pr(y^n \mid {\bf w}, {\bf s}^n, {\bf x}^n) \Pr\big( \widehat{\bf W}={\bf w}\mid {\bf w}, {\bf s}^n, {\bf x}^n, y^n \big) \\
	&=\sum_{{\bf x}^n}\Pr\big( {\bf x}^n \mid {\bf w},  {\bf s}^n)  \sum_{y^n\colon \psi(y^n)= {\bf w}} \prod_{i=1}^n \mN(y_i\mid {\bf x}_i, {\bf s}_i)\label{eq:prob_error_EA_intermediate_1}
\end{align}
To justify the last step, first we note that since ${\bf W} \leftrightarrow  ({\bf S}^n, {\bf X}^n) \leftrightarrow {\bf Y}^n$ form a Markov chain and the channel is memoryless,  we have $\Pr(y^n \mid {\bf w}, {\bf s}^n, {\bf x}^n) = \Pr(y^n \mid  {\bf s}^n, {\bf x}^n) = \prod_{i=1}^n \mN(y_i\mid {\bf x}_i, {\bf s}_i)$. Then, note that $\Pr\big( \widehat{\bf W}={\bf w}\mid {\bf w}, {\bf s}^n, {\bf x}^n, y^n \big) = 1$ if the decoded messages $\psi(y^n)={\bf w}$, and it equals $0$ otherwise. 

The term $\Pr\big( {\bf x}^n \mid {\bf w}, {\bf s}^n \big)$ depends on the encoding quantum instruments. To state it explicitly, for each $k\in [K]$, let 
\begin{equation*}
	\overline{\mathcal{E}}^{(k)}_{x_k^n\mid (w_k,s_k^n)} \triangleq \mathcal{E}_{x_{k,n}\mid (w_k,s_k^n,x_k^{n-1})}^{(k,n)} \circ \mathcal{E}_{x_{k,n-1}\mid (w_k,s_k^{n-1},x_k^{n-2})}^{(k,n-1)}  \circ \cdots \circ \mathcal{E}_{x_{k,1}\mid (w_k, s_{k,1})}^{(k,1)}
\end{equation*}
denote the overall encoding quantum instrument at Transmitter $k$, conditioned on the message  and the channel states across $n$ channel uses, where $\circ$ denotes the composition of two maps. Then
\begin{align}
	\Pr\big( {\bf x}^n \mid {\bf w}, {\bf s}^n \big) = \Tr \Big[ \big( \underbrace{\overline{\mathcal{E}}_{x_1^n\mid (w_1, s_1^n)}^{(1)} \otimes \cdots \otimes \overline{\mathcal{E}}_{x_K^n\mid (w_K,s_K^n)}^{(K)}\big)(\rho)}_{\triangleq \widetilde{\rho}_{{\bf x}^n, {\bf w}, {\bf s}^n} } \Big].
\end{align}
Note that $\widetilde{\rho}_{{\bf x}^n, {\bf w}, {\bf s}^n}$ is an unnormalized density operator. The normalized version is $\widetilde{\rho}_{{\bf x}^n, {\bf w}, {\bf s}^n}/\Pr\big( {\bf x}^n \mid {\bf w}, {\bf s}^n \big)$.

In either case, plugging $\Pr\big(\widehat{\bf W}={\bf w} \mid {\bf W}={\bf w}, {\bf S}^n={\bf s}^n)$ into \eqref{eq:prob_error_first} yields an explicit expression for $1-P_e(\Gamma)$. \hfil \qed

\section{Proof of Theorem \ref{thm:additive_MACs_classical}}\label{proof:additive_MACs_classical}
\begin{proof}[Proof of classical achievability]
Let $t_1,t_2,\ldots,t_K$ be  functions that minimize the objective in \eqref{eq:min_obj}. For each $k\in [K]$, at each channel use, Transmitter $k$ observes $S_k$ and sends $X_k = X_k' \oplus_m t_k(S_k)$ where $X_k' \in [0:m-1]$. The resulting channel output is $Y = (X_1'\oplus_m t_1(S_1)) \oplus_m \cdots \oplus_m (X_K'\oplus_m t_K(S_K)) \oplus_m Z = X_1'\oplus_m \cdots \oplus_m X_K' \oplus_m Z'$ where $Z' \triangleq t_1(S_1) \oplus_m \cdots \oplus_m t_K(S_K) \oplus_m Z$ and we have $H(Z') = H_{\min}$. Therefore, we have obtained a converted channel $Y = X_1'\oplus_m \cdots \oplus_m X_K' \oplus_m Z'$ with $H(Z') = H_{\min}$ and $Z'$ independent of $(X_1',\ldots,X_K')$. In the converted channel, if we set $X_k'=0$ for $k\in \{2,\ldots, K\}$, then we have a point-to-point channel (from Transmitter $1$ to the receiver) with $Y = X_1'\oplus_m Z'$. It follows from the direct coding theorem of a point-to-point channel that, by choosing $X_1'$ uniformly over $[0:m-1]$, any rate $R < I(X_1';X_1' \oplus_m Z') = H(X_1' \oplus_m Z') - H(X_1' \oplus_m Z' \mid X_1') = \log(m) - H(Z') = \log_2(m) - H_{\min}$ is achievable for message $W_1$. By symmetry and a standard time-division multiple access (TDMA) strategy, any rate tuple satisfying $R_1+\cdots +R_K < \log_2(m) - H_{\min}$ is achievable.
\end{proof}

\begin{proof}[Proof of classical converse]
Given any classical coding schemes $\{\Gamma_n\}$ with $\lim_{n\to \infty} P_e(\Gamma_n) = 0$ and $\lim_{n\to\infty}\log_2(M_{k,n})/n \geq R_k,\forall k\in [K]$, recall that $J$ denotes the shared random variable that is independent of the messages $(W_1,\ldots, W_K)$. We have
\begin{align}
	&n(R_1+\cdots +R_K) -o(n) \notag \\
	&\leq I(\underbrace{W_1,\ldots, W_K}_{\triangleq {\bf W}};Y^n \mid J) \label{eq:use_Fano} \\
	&= H(Y^n\mid J) - \sum_{i=1}^n H(Y_i\mid {\bf W}, Y^{i-1},J) \label{eq:use_chain} \\
	&\leq n\log_2(m) - \sum_{i=1}^n H(Y_i\mid {\bf W}, Y^{i-1}, S_1^{i-1},\ldots, S_{K}^{i-1},J) \label{eq:use_CRE} \\
	&= n\log_2(m) - \sum_{i=1}^n H(Y_i\mid {\bf W},  S_1^{i-1},\ldots, S_{K}^{i-1},J)\label{eq:use_encoder} 
\end{align}
where Step \eqref{eq:use_Fano} follows from Fano's inequality and the fact that the shared randomness $J$  is independent of the messages $(W_1,\ldots, W_K)$. Step \eqref{eq:use_chain} uses the definition of conditional mutual information and the chain rule for conditional entropy. Step \eqref{eq:use_CRE} follows as conditioning does not increase entropy. Step \eqref{eq:use_encoder} is because of the Markov chain $Y^{i-1}\leftrightarrow ({\bf W},S_1^{i-1},\ldots, S_{K}^{i-1},J ) \leftrightarrow Y_i$ since the channel is memoryless. 

Now, for any fixed realization of $({\bf W},  S_1^{i-1},\ldots, S_{K}^{i-1},J)$, we note that for each $k\in [K]$, the input at Transmitter $k$ at the $i^{th}$ channel use is $X_{k,i} = t_k(S_{k,i})$ for some function $t_k\colon \mathcal{S}_k \to \mathcal{X}_k$. Also note that $(S_{1,i}, S_{2,i}, \ldots, S_{K,i}, Z_i)$ are independent of $({\bf W},  S_1^{i-1},\ldots, S_{K}^{i-1},J)$. Therefore, for each $i\in [n]$,
\begin{align}
	&H(Y_i\mid {\bf W},  S_1^{i-1},\ldots, S_{K}^{i-1},J) \notag \\
	&\geq \min_{t_1,\ldots, t_K}H\big(t_1(S_{1,i})\oplus_m \cdots \oplus_m t_K(S_{K,i})\oplus_m Z_i \big)\\
	&= H_{\min}
\end{align}
From \eqref{eq:use_encoder}, we have  $n(R_1+\cdots +R_K)-o(n) \leq n\log_2(m)-nH_{\min}$. Thus, we conclude that $C^{\C}\leq \log_2(m) - H_{\min}$. 	
\end{proof}

\section{Proof of Theorem \ref{thm:EAIS}} \label{proof:EAIS}
Given $\rho\in \mathcal{D}(\mathcal{H}_{Q_1} \otimes \cdots \otimes \mathcal{H}_{Q_K})$, for each channel use, let the $K$ transmitters share the entangled resource in the state $\rho$, such that Transmitter $k$ has $Q_k$ for each $k\in [K]$. For each $k\in [K]$, Transmitter $k$ measures $Q_k$ on the POVM $\mathcal{F}^{(k)}_{s_k} = \{F^{(k)}_{b_k\mid s_k}\}_{b_k\in [0:m-1]}$ when the channel state $S_k=s_k$, and obtains the outcome $B_k$. It then sends $X_k=X_k'\oplus_m B_k$ for some $X_k' \in \mathcal{X}_k$. The receiver sees $Y = X_1\oplus_m \cdots \oplus_m X_k\oplus_m Z = X_1'\oplus_m \cdots \oplus_m X_k' \oplus_m ( B_1\oplus_m \cdots  \oplus_m B_K \oplus_m Z)$. Denote $Z'\triangleq ( B_1\oplus_m \cdots  \oplus_m B_K \oplus_m Z)$. It follows that the converted channel $(X_1', \ldots, X_K') \to Y$ is another additive MAC with state-dependent interference term equal to $Z'$.  In the converted channel, if we set $X_k'=0$ for $k\in \{2,\ldots, K\}$, then we have a point-to-point channel (from Transmitter $1$ to the receiver) with $Y = X_1'\oplus_m Z'$. It follows from the direct coding theorem of a point-to-point channel that, by choosing $X_1'$ uniformly over $[0:m-1]$, any rate $R < I(X_1';X_1' \oplus_m Z') = H(X_1' \oplus_m Z') - H(X_1' \oplus_m Z' \mid X_1') = \log_2(m) - H(Z')$ is achievable for message $W_1$. By symmetry and a standard time-division multiple access (TDMA) strategy, any rate tuple satisfying $R_1+\cdots +R_K<\log_2(m)-H(Z')$ is achievable. \hfil \qed

\section{Proofs of Propositions 3, 4, and justification for Remark 4} 
\subsection{Proof of Proposition \ref{prop:ChannelB1}} \label{proof:MerminGHZ}
\begin{proof}[Proof of classical capacity of Channel B1] 
To obtain $C^{\C}$, we apply Theorem \ref{thm:additive_MACs_classical}. The argument to find $H_{\min}$ is as follows. Let us focus on the objective in \eqref{eq:min_obj}. 
Let $V \triangleq t_1(S_1)\oplus \cdots \oplus t_K(S_K) \oplus Z$. We then have $H_{\min} = \min_{t_1,\ldots, t_K} H(V)$. Let $\Delta \triangleq \Pr(V=0) - \Pr(V=1)$, and thus $H(V) = H_b(\frac{1+|\Delta|}{2})$.
Let $\Sigma_0 \triangleq \{(s_1,\ldots, s_K)\in \{0,1\}^K\colon \sum_{k=1}^K s_k \pmod 2 = 0\}$. Next let us bound $\Delta$ following essentially a result of Mermin \cite{merminGHZ}. A detailed proof of this bound has appeared in \cite{brassard2004recasting, watts2019exponential} as well. We  have,
\begin{align}
	&\Delta =\mathbb{E}[(-1)^V]  \\
	&=\frac{1}{2^{K-1}}\sum_{s^K\in \Sigma_0} \mathbb{E}[(-1)^V\mid S^K=s^K]\\
	&= \frac{1}{2^{K-1}}\sum_{s^K\in \Sigma_0}  (-1)^{(s_1+\cdots + s_K)/2} \prod_{k=1}^K (-1)^{t_k(s_k)}\\
	&=\frac{1}{2^{K-1}} {\rm Re} \Big\{ \sum_{s^K\in \{0,1\}^K} \prod_{k=1}^K i^{s_k} (-1)^{t_k(s_k)} \Big\} \label{eq:intermediate_1}
\end{align}
For each $k\in [K]$, since $s_k\in \{0,1\}$, $(-1)^{t_k(s_k)} = (-1)^{t_k(0)}\times r_k^{s_k}$ where $r_k \triangleq \tfrac{(-1)^{t_k(1)}}{(-1)^{t_k(0)}} \in \{1,-1\}$. Therefore,
\begin{align}
	&\eqref{eq:intermediate_1} = \frac{1}{2^{K-1}} \Big(\prod_{k=1}^K (-1)^{t_k(0)}\Big) {\rm Re} \Big\{ \sum_{s^K\in \{0,1\}^K} \prod_{k=1}^K (i r_k)^{s_k} \Big\}\\
	&=\frac{\pm 1}{2^{K-1}}{\rm Re}\Big\{ \prod_{k=1}^K \sum_{s\in \{0,1\}} (ir_k)^{s} \Big\} \\
	&= \frac{\pm 1}{2^{K-1}}{\rm Re}\Big\{ \prod_{k=1}^K (1+ i r_k) \Big\} \\
	&= \frac{\pm \sqrt{2}^K}{2^{K-1}}{\rm Re}\Big\{ \prod_{k=1}^K e^{ir_k\pi/4 } \Big\}
\end{align}
It follows that when $K$ is even, $|\Delta| \leq \frac{\sqrt{2}^K}{2^{K-1}}$, with equality when $r_k=(-1)^{k}$. When $K$ is odd,  $|\Delta| \leq \frac{\sqrt{2}^{K}}{2^{K-1}}\times \frac{\sqrt{2}}{2}$, with equality when $r_k=1$. Therefore, $|\Delta| \leq 2^{1-\lceil K/2 \rceil}$. Since in both cases the equality can be achieved, we have that $H_{\min} = \min_{t_1,\ldots,t_K}H(V) = H_b(\frac{1}{2}+2^{-\lceil K/2 \rceil})$, and thus $C^{\C} = 1-H_b(\frac{1}{2}+2^{-\lceil K/2 \rceil})$.
\end{proof}

\begin{proof}[Proof of EA capacity of Channel B1]
Let ${\sf X}$ be the Pauli X operator (gate) such that ${\sf X}\ket{0} = \ket{1}$ and ${\sf X}\ket{1} = \ket{0}$. Let ${\sf R}$ be the phase gate such that ${\sf R}\ket{0} = \ket{0}$ and ${\sf R}\ket{1} = i \ket{1}$. 
The scheme that achieves $C^{\EA}=1$ uses a channel conversion strategy combined with the strategy that wins the Mermin-GHZ game \cite{merminGHZ}. For each channel use, the $K$ transmitters convert the channel $Y=X_1\oplus \cdots \oplus X_K \oplus Z$ to an interference-free channel $Y=X_1'\oplus \cdots \oplus X_K'$, where $X_k'\in \{0,1\}$ for each $k\in [K]$, by performing individual state-dependent measurements on their respective entangled qubits. Specifically, for each channel use, the $K$ transmitters consume $K$ qubits, denoted by $Q_1\cdots Q_K$, that are set in the entangled (GHZ) state $\ket{\Phi^+} \triangleq \frac{\sqrt{2}}{2}(\ket{0}^{\otimes K}+\ket{1}^{\otimes K})$. For each $k\in [K]$, $Q_k$ is with Transmitter $k$. Transmitter $k$ applies the phase gate ${\sf R}$ to $Q_k$ if $S_k=1$. It then measures $Q_k$ on the Pauli X operator. Denote by $\widetilde{B}_k \in \{1,-1\}$ the measurement outcome on $Q_k$, and set $B_k=0$ if $\widetilde{B}_k=1$ and $B_k=1$ if $\widetilde{B}_k=-1$. Transmitter $k$ then sends to the channel $X_k = X'_k\oplus B_k$, where $X_k'\in \{0, 1\}$ denotes the input of the converted channel. 
To analyze the strategy, let $\ket{\Psi}$ denote the state of $Q_1\cdots Q_K$ after the transmitters have applied the conditional phase gates. We have that $\ket{\Psi} = \frac{\sqrt{2}}{2}(\ket{0}^{\otimes K}+i^{S_1+\cdots +S_K}\ket{1}^{\otimes K})$.
If $Z = (S_1+\cdots +S_K)/2 \pmod 2 = 0$, we have $\ket{\Psi} = \ket{\Phi^+}$.  On the other hand, if $Z = 1$, we have $\ket{\Psi} = \frac{\sqrt{2}}{2}(\ket{0}^{\otimes K}-\ket{1}^{\otimes K})\triangleq \ket{\Phi^-}$.
Since $\ket{\Phi^+}$ is a $(+1)$-eigenstate for the operator ${\sf X}^{\otimes K}$ and $\ket{\Phi^-}$ is a $(-1)$-eigenstate for ${\sf X}^{\otimes K}$, it follows that $\Pr(B_1\oplus\cdots \oplus B_K=0\mid Z = 0) = \Pr(\widetilde{B}_1\cdots \widetilde{B}_K=1\mid Z = 0) = |\braket{\Phi^+|\Phi^+}|^2=1$, and similarly $\Pr(B_1\oplus\cdots \oplus B_K=1 \mid Z = 1) = \Pr(\widetilde{B}_1\cdots \widetilde{B}_K=-1\mid Z = 1) = |\braket{\Phi^-|\Phi^-}|^2=1$. Therefore, $\Pr(B_1\oplus\cdots \oplus B_K \oplus Z = 0) = 1$. Since $Y = X_1\oplus \cdots \oplus X_K \oplus Z = (X_1'\oplus \cdots \oplus X_K') \oplus (B_1\oplus\cdots \oplus B_K \oplus Z)$, we have that $Y =  X_1'\oplus \cdots \oplus X_K'$ with certainty. Over this converted channel, any rate tuple that satisfies $R_1+\cdots +R_K\leq 1$ can be achieved by a simple scheduling scheme.

On the other hand, even if all $K$ transmitters cooperate,  they cannot\footnote{For a point-to-point channel with state, \scalebox{0.95}{$(\mN_{Y\mid XS}, \mP_S)$}, the entanglement-assisted (in fact even non-signaling assisted) capacity is upper bounded by \scalebox{0.95}{$\max_{\mP_{\!X\mid S}}I(X;Y\mid S)$} (see \cite{Yao_Jafar_CSITTP}), which evaluates to $1$ in the example.} achieve a rate larger than $1$.  Thus, $C^{\EA} = 1$.	
\end{proof}

\subsection{Proof of Proposition \ref{prop:ChannelB2}} \label{proof:Mermin_indep}
For \emph{Channel B2}, $Z = \lceil (S_1+\cdots + S_K)/2 \rceil \pmod 2 = (S_1+\cdots + S_K+S_{0})  /2 \pmod 2$, where we defined $S_0 \triangleq \bigoplus_{k=1}^K S_k$. To show that $C^{\C} \leq 1-H_b(\tfrac{1}{2}+2^{\lceil (K+1)/2\rceil})$, note that for Channel B2, $H_{\min} = \min_{t_1,\ldots, t_K}H(\bigoplus_{k=1}^K t_k(S_k) \oplus Z) \geq \min_{t_0,t_1,\ldots, t_K}H(\bigoplus_{k=0}^K t_k(S_k) \oplus Z) = H_b(\frac{1}{2}+2^{-\lceil {(K+1)}/2 \rceil})$, where the last step follows the same lines of the proof of the classical capacity of Channel B1 by considering $K+1$ users. Therefore, $C^{\C}\leq 1-H_b(\frac{1}{2}+2^{-\lceil {(K+1)}/2 \rceil})$. We again have $C^{\C}=O(2^{-K})$ by Lemma \ref{lem:approximation}.

With entanglement assistance, we again use the same channel conversion strategy. Note that conditioned on $S_0=0$, which happens with probability $1/2$, Channel B2 becomes Channel B1, so the interference term is completely canceled by the channel conversion strategy. Conditioned on $S_0=1$, after the transmitters apply their conditional phase gates, the state of the qubits $\ket{\Psi}$ is either equal to $\frac{\sqrt{2}}{2}(\ket{0}^{\otimes K} + i \ket{1}^{\otimes K}) = \frac{1}{2}\big((1+i) \ket{\Phi^+} + (1-i) \ket{\Phi^-} \big)$, or equal to $\frac{\sqrt{2}}{2}(\ket{0}^{\otimes K} - i \ket{1}^{\otimes K}) = \frac{1}{2}\big((1-i) \ket{\Phi^+} + (1+i) \ket{\Phi^-} \big)$. Note that in either case, $\ket{\Psi}$ is in the span of $\{\ket{\Phi^+}, \ket{\Phi^-}\}$. 
Also, for all cases where $S_0=1$, we have $|\braket{\Phi^+|\Psi}|^2 = |\braket{\Phi^-|\Psi}|^2 = 1/2$, regardless of the value of $Z$. 
It follows that $\Pr(B_1\oplus \cdots \oplus B_K \oplus Z = z \mid S_0=1) = 1/2, \forall z\in \{0,1\}$. Recall that $Y = (X_1'\oplus \cdots \oplus X_K') \oplus (B_1\oplus\cdots \oplus B_K \oplus Z)$. Observe that conditioned on $S_0=1$, the converted channel is completely noisy, in that the output is uniformly distributed over $\{0,1\}$, regardless of the inputs $X_1',\ldots, X_K'$. Overall, we have   $Y = X_1'\oplus \cdots \oplus X_K'\oplus Z'$ where $Z'$ is an independent noise term with $\Pr(Z'=0) = 3/4$ and $\Pr(Z'=1) = 1/4$. Thus, any rate tuple  satisfying $R_1+\cdots+R_K<1-H_b(3/4)$ is achievable with entanglement assistance. \hfil \qed

\subsection{Justification for Remark \ref{rem:robust} (Robustness)} \label{proof:robust}
The above entanglement-assisted scheme assumes that the qubits consumed for each channel use are in the perfect GHZ state $\ket{\Phi^+}$. In contrast, suppose that each qubit (independently) has a  probability $1-\gamma$ of being completely depolarized, where $\gamma \in [0,1]$. Consider any use of Channel B1 (or Channel B2). 
Let $D$ denote the event that ``at least one qubit in the GHZ state is depolarized," and let $\overline{D}$ denote the complement of $D$, i.e., ``all $K$ qubits are intact." Then we have $\Pr(D) = 1-\gamma^K$ and $\Pr(\overline{D}) = \gamma^K$. It is not difficult to see that the depolarization of any qubit (say $Q_k$)  causes $B_k$ to be a random noise uniformly distributed over $\{0,1\}$. It follows that $\Pr(Y=X_1'\oplus \cdots \oplus X_K' \mid D) = \frac{1}{2}$. Thus, $\Pr(Y=X_1'\oplus \cdots \oplus X_K') = \Pr(Y=X_1'\oplus \cdots \oplus X_K' \mid D)\Pr(D) + \Pr(Y=X_1'\oplus \cdots \oplus X_K' \mid \overline{D})\Pr(\overline{D}) = \frac{1}{2}(1-\gamma^{K})+\alpha \gamma^{K} = \frac{1}{2}+(\alpha-\frac{1}{2})\gamma^K$, where $\alpha=1$ for Channel B1, and $\alpha = \frac{3}{4}$ for Channel B2. 
Equivalently,  write the converted channel as $Y = X_1'\oplus \cdots \oplus X_K'\oplus Z'$ where $Z'$ is an independent noise term with $\Pr(Z'=0) = \frac{1}{2}+(\alpha-\frac{1}{2})\gamma^K$ and $\Pr(Z'=1) = \frac{1}{2}-(\alpha-\frac{1}{2})\gamma^K$. Thus, any rate tuple satisfying $R_1+\cdots+R_K<1-H_b\big(\frac{1}{2}+(\alpha-\frac{1}{2})\gamma^K\big)$ is achievable, via a binary symmetric channel (BSC) code for each message along with TDMA over this converted channel. The rate thus achieved  is $\Omega(\gamma^{2K})$ as $K \to \infty$ by Lemma \ref{lem:approximation}.

Recall from Proposition \ref{prop:ChannelB1} and Proposition \ref{prop:ChannelB2} that $C^{\C} = O\big((\frac{1}{2})^K\big)$, we conclude that as long as $\gamma^2 > \frac{1}{2}$, i.e., $\gamma > \frac{\sqrt{2}}{2} \approx 70\%$, there is still an exponential capacity gain from entanglement assistance for Channel B1 and Channel B2, and to achieve such a gain it is sufficient (albeit not necessarily optimal) to directly apply the channel conversion strategy  designed for the case of perfect entanglement. \hfil \qed

\section{Proofs of Propositions \ref{prop:ClassA}, \ref{prop:ClassB} and \ref{prop:ClassC}} \label{proof:use_xor_games}
The proofs in this section require results on XOR games\footnote{XOR games belong to a special class of non-local games, which serve as useful frameworks for studying quantum and general nonlocality.} \cite{ambainis2012quantum, ambainis2012advantage, ambainis2013provable, Briet_and_Vidick}. For our purpose, let us consider the following definition of XOR games in Definition \ref{def:XOR_game}.
\begin{definition}[Multiplayer XOR game] \label{def:XOR_game}
	Let $K\geq 2$. A $K$-player XOR game is specified by a tuple $(\{\mathcal{A}_k\}_{k\in [K]},\mP_{\!A_1\cdots A_K},f)$ (known to everyone). There are $K$ non-communicating players, each is given an input variable ($A_k\in \mathcal{A}_k$ for the $k^{th}$ player), and is asked to produce an output variable ($B_k\in \{0,1\}$ for the $k^{th}$ player). The prior distribution of $(A_1,A_2,\ldots, A_K)$ is $\mP_{\!A_1A_2\cdots A_K}$. The winning condition is defined by a function $f \colon \mathcal{A}_1\times \cdots \times \mathcal{A}_K \to \{0,1\}$, such that the players win the game if $B_1  \oplus \cdots \oplus B_K = f(A_1,\ldots, A_K)$.
\end{definition}
 
Given an XOR game, we study the classical strategies and the entanglement-assisted (EA) strategies that win the game (with certain probability). The definitions for these strategies are given as follows.
\begin{definition}[Classical game strategy] \label{def:classical_game_strategy}
	In a classical game strategy, Player $k$ produces $B_k = t_k(A_k)$ where $t_k\colon \mathcal{A}_k\to \mathcal{B}_k$ is a function. The strategy should specify $t_k$ for Player $k$. Denote by $\omega^{\C}$ the optimal probability of winning the game using classical strategies. Let us call $\beta^{\C} \triangleq \omega^{\C} - (1-\omega^{\C}) = 2\omega^{\C} -1$ the classical bias.\footnote{The players may generate their answers using additional shared randomness; however, this cannot improve the optimal winning probability.}
\end{definition}
\begin{definition}[EA game strategy]
	In an entanglement-assisted game strategy, the $K$ players share an (entangled) $K$-partite quantum resource $Q_1\cdots Q_K$ in the state $\rho$. Player $k$ produces $B_k$ by conducting a measurement on $Q_k$ depending on $A_k$. The strategy should specify the measurement conditioned on $A_k=a_k$ for each $a_k\in \mathcal{A}_k$ (e.g., a POVM $\mathcal{F}^{(k)}_{a_k}= \{F^{(k)}_{b_k \mid a_k}\}_{b_k \in \{0,1\}}$ with binary outputs in $\{0,1\}$)  for Player $k$.  Denote by $\omega^{\EA}$ the optimal probability of winning the game using EA strategies. Let us call $\beta^{\EA} \triangleq \omega^{\EA} - (1-\omega^{\EA}) = 2\omega^{\EA} -1$ the entanglement-assisted bias.
\end{definition}

Now, let us make a connection from any XOR game to an $(m=2)$ additive MACs with state-dependent interference.
Given the XOR game $(\{\mathcal{A}_k\}_{k\in [K]},\mP_{\!A_1\cdots A_K},f)$, consider the following additive MAC with state-dependent interference: the state $S_k$ of the MAC is equal to $A_k$ of the game, and the interference term $Z$ is equal to $f(S_1,\ldots, S_K)$. It follows that the channel has the form
\begin{align} \label{eq:channel_from_XOR_game}
	Y = X_1\oplus X_2 \oplus \cdots \oplus X_K \oplus f(S_1,\ldots, S_K).
\end{align}
We then have the following lemmas, which connect the winning probability of a game to the capacity of the corresponding state-dependent MAC.
\begin{lemma} \label{lem:equivalence_classical}
	If the XOR game has optimal classical winning probability equal to $\omega^{\C}$,  then the corresponding channel defined in \eqref{eq:channel_from_XOR_game} has $C^{\C} = 1-H_b(\omega^{\C}) = 1-H_b(\frac{1}{2}+\frac{\beta^{\C}}{2})$.
\end{lemma}
\begin{lemma} \label{lem:achievability_by_XOR}
	If the XOR game has optimal entanglement-assisted winning probability equal to $\omega^{\EA}$, then the corresponding channel defined in \eqref{eq:channel_from_XOR_game} has   $C^{\EA} \geq 1-H_b(\omega^{\EA})=1-H_b(\frac{1}{2}+\frac{\beta^{\EA}}{2})$.
\end{lemma}
The proofs of Lemma \ref{lem:equivalence_classical} and Lemma \ref{lem:achievability_by_XOR} are rather straightforward and are therefore deferred to the end of this section.
We are now ready to prove Propositions \ref{prop:ClassA}, \ref{prop:ClassB} and \ref{prop:ClassC}.

\begin{proof}[Proof of Proposition \ref{prop:ClassA}]
For $L \geq 2$, consider the class of $2$-player XOR games studied in \cite{ambainis2012quantum}, where $\mathcal{A}_1 = \mathcal{A}_2 = [L]$, $\mP_{\!A_1A_2}(a_1,a_2) = 1/L^2, \forall a_1, a_2\in [L]^2$. Suppose one randomly picks a game $\gm_L$ within this class by randomly choosing the value $f(a_1,a_2)$ equally likely in $\{0,1\}$, independently for each $\forall (a_1, a_2) \in [L]^2$. Then \cite[Thm. 1]{ambainis2012quantum} implies that with probability $1-o(1)$, the picked game has $\beta^{\EA}(\gm_L) = \frac{2+o(1)}{\sqrt{L}}$ as $L\to \infty$, and \cite[Thm. 4]{ambainis2012quantum} implies that with probability $1-o(1)$, the picked game has $\beta^{\C}(\gm_L) \leq \frac{2\sqrt{\ln(2)}+o(1)}{\sqrt{L}}$ as $L \to \infty$. 
Note that randomly picking the game $\gm_L$ in this way is equivalently to randomly picking the channel $\ch_L$ in Class A as described in Proposition \ref{prop:ClassA}.
According to Lemma \ref{lem:achievability_by_XOR} and Lemma \ref{lem:equivalence_classical}, we have that $C^{\EA}(\ch_L) \geq 1-H_b(\frac{1}{2} + \frac{1+o(1)}{\sqrt{L}}) \stackrel{(a)}{=} \frac{2}{\ln (2) }L^{-1} + o(L^{-1})$, and that $C^{\C}(\ch_L) \leq 1-H_b(\frac{1}{2} + \frac{\sqrt{\ln(2)}+o(1)}{\sqrt{L}}) \stackrel{(a)}{=} 2L^{-1} + o(L^{-1})$, both with probability $1-o(1)$, where steps labeled by $(a)$ follow from Lemma \ref{lem:approximation}.
\end{proof}
\vspace{0.1cm}
\begin{proof}[Proof of Proposition \ref{prop:ClassB}]
For $K\geq 2$, consider the class of $K$-player XOR games studied in \cite{ambainis2013provable}, where $\mathcal{A}_k=\{0,1\}, \forall k\in [K]$, $\mP_{\!A_1\cdots A_K}(a_1,\ldots, a_K) = 2^{-K}, \forall (a_1,\ldots, a_K)\in \{0,1\}^K$. Suppose one randomly picks a game $\gm_K$ within this class by randomly choosing the values $(G_0,G_1,\ldots, G_K)\in \{0,1\}^{K+1}$ such that $f(a_1,\ldots, a_K) = G_{a_1+\cdots +a_K}, \forall (a_1,\ldots, a_K)\in \{0,1\}^K$. Then \cite[Cor. 4]{ambainis2013provable} implies that with probability $1-o(1)$, the picked game has $\beta^{\EA}(\gm_K) = \Omega(\frac{\sqrt{\log(K)}}{K^{1/4}})$, and \cite[Thm. 4]{ambainis2012advantage} (see also \cite{ambainis2013provable}) implies that with probability $1-o(1)$, the picked game has $\beta^{\C}(\gm_K)=O(\frac{1}{K^{1/4}})$. Note that randomly picking the game $\gm_K$ in this way is equivalent to randomly picking the channel $\ch_K$ in Class B as described in Proposition \ref{prop:ClassB}. According to Lemma \ref{lem:achievability_by_XOR} and Lemma \ref{lem:equivalence_classical}, we have that $C^{\EA}(\ch_K) \geq 1- H_b\big(\frac{1}{2}+\Omega(\frac{\sqrt{\log(K)}}{K^{1/4}})\big) \stackrel{(a)}{=} \Omega(\frac{\log(K)}{K^{1/2}})$, and that $C^{\C}(\ch_K) \leq 1- H_b\big(\frac{1}{2}+O(\frac{1}{K^{1/4}})\big) \stackrel{(a)}{=} O(\frac{1}{K^{1/2}})$, both with probability $1-o(1)$, where steps labeled by $(a)$ follow from Lemma \ref{lem:approximation}.	
\end{proof}
\vspace{0.1cm}
\begin{proof}[Proof of Proposition \ref{prop:ClassC}]
We use the result of \cite[Thm. 1]{Briet_and_Vidick}. It implies that for $r\in \mathbb{N}$ and $L=2^r$, there exists a sequence of $3$-player XOR games $\{\gm_r\}_r$, where $\gm_r$ has  $|\mathcal{A}_1|=|\mathcal{A}_2|=|\mathcal{A}_3|=L^2$, such that $\beta^{\EA}(\gm_r) \geq \Omega\big(\sqrt{L}\log^{-5/2}(L)\big)\beta^{\C}(\gm_r)$. Since $\beta^{\EA}(\gm_r) = 2\omega^{\EA}(\gm_r)-1 \leq 1$, it follows that $\beta^{\C}(\gm_r) = O(\frac{\log^{5/2}(L)}{\sqrt{L}})$, which converges to $0$ as $r\to \infty$.

Consider the corresponding $3$-user MACs $\{\ch_r\}_r$ associated with $\{\gm_r\}_r$. Since $\beta^{\C}(\gm_r) \to 0$ as $r\to \infty$, Lemma \ref{lem:equivalence_classical} then implies that $C^{\C}(\ch_r) = \big(\frac{2}{\ln(2)}+o(1)\big)\big(\beta^{\C}(\gm_r)\big)^2$ as $r\to \infty$. 
On the other hand, Lemma \ref{lem:achievability_by_XOR} implies that $C^{\EA}(\ch_r)\geq 1-H_b\big(\frac{1}{2}+\frac{\beta^{\EA}(\gm_r)}{2}\big) \geq \frac{1}{2\ln(2)}\big(\beta^{\EA}(\gm_r)\big)^2$, where the last inequality follows from Lemma \ref{lem:approximation}. Combining with the result of \cite[Thm. 1]{Briet_and_Vidick}, we conclude that $C^{\EA}(\ch_r)/C^{\C}(\ch_r) = \Omega\big( L \log^{-5}(L) \big) = \Omega(2^{r} r^{-5})$.
\end{proof}

\begin{proof}[Proof of Lemma \ref{lem:equivalence_classical}]
	Applying Theorem \ref{thm:additive_MACs_classical} to the channel in \eqref{eq:channel_from_XOR_game}, we have
	\begin{align}
		C^{\C} &=  1- \min_{t_1,\ldots, t_K} H\big( \underbrace{t_1(A_1) \oplus \cdots \oplus t_K(A_K)\oplus f(A_1,\ldots, A_K)}_{\triangleq V} \big) \\
		& = \max_{t_1,\ldots, t_K} \big( 1- H_b\big( \Pr(V=0) \big) \big) \label{eq:obj_entropy} \\
		& = 1- H_b\big(\max_{t_1,\ldots,t_K}\Pr(V=0)\big) \label{eq:use_monotonicity}
	\end{align}
	where the maximization is over all functions $(t_1,\cdots, t_K)$ such that $t_k\colon \mathcal{A}_k \to \{0,1\}, \forall k\in[K]$, and we define $V(t_1,\ldots, t_K) \triangleq t_1(A_1) \oplus \cdots \oplus t_K(A_K)\oplus f(A_1,\ldots, A_K)$   which depends on the choices of $(t_1,\cdots, t_K)$. We claim that it suffices to consider $(t_1,\cdots, t_K)$ such that $\Pr(V=0)\geq 1/2$. This is because if the maximum in \eqref{eq:obj_entropy} is attained by some $(t_1,\ldots, t_K)$ which yields $\Pr\big(V(t_1,t_2,\ldots, t_K) =0 \big)<1/2$, then the altered choice $(\tilde{t}_1,t_2,\ldots, t_K)$ where $\tilde{t}_1(a_1) \triangleq t_1(a_1) \oplus 1$ for $a_1\in \mathcal{A}_1$ yields $\Pr\big( V(\tilde{t}_1,t_2,\ldots, t_K) = 0\big) = 1- \Pr\big(V(t_1,t_2,\ldots, t_K) =0 \big) > 1/2$, and it also achieve the maximum value in \eqref{eq:obj_entropy} because $1-H_b(p)$ is symmetric about $p=1/2$. Step \eqref{eq:use_monotonicity} follows as the function $1-H_b(p)$ is monotonically increasing in $p\in [1/2, 1]$ (see also Fig. \ref{fig:Hbp_complement}).
	
	Meanwhile, according to the definition of classical strategies (Definition \ref{def:classical_game_strategy}), 
	\begin{align}
		\omega^{\C} &=\max_{t_1,\ldots, t_K}\Pr\big(t_1(A_1) \oplus \cdots \oplus t_K(A_K) = f(A_1,\ldots, A_K)\big) \\
		&= \max_{t_1,\ldots, t_K}\Pr(V = 0) \label{eq:obj_prob}
	\end{align}
	and note that it also suffices to consider $(t_1,\ldots, t_K)$ for which $\Pr(V = 0)\geq 1/2$ because $\omega^{\C} \geq 1/2$. Combining \eqref{eq:use_monotonicity} and \eqref{eq:obj_prob}, we conclude that $C^{\C} = 1-H_b(\omega^{\C})$.
\end{proof}

\begin{proof}[Proof of Lemma \ref{lem:achievability_by_XOR}]
	This lemma essentially follows from Theorem \ref{thm:EAIS}. Suppose we are given the EA strategy that wins the game with probability $\omega^{\EA}$. For each $k\in [K]$, Transmitter $k$ observes the channel state $S_k$ and treats it as the input $A_k$ in the game. Transmitter $k$ then sends $X_k = X_k'\oplus B_k$, where $X_k'\in \{0,1\}$, and $B_k$ is the output associated with Player $k$ according to the strategy of the game. The receiver  sees $Y = (X_1'\oplus B_1) \oplus \cdots \oplus (X_K'\oplus B_K) \oplus f(S_1,\ldots, S_K) = (X_1'\oplus\cdots \oplus X_K') \oplus  (B_1 \oplus \cdots \oplus B_K \oplus f(S_1,\ldots, S_K))$. Since $\Pr(B_1\oplus \cdots \oplus B_K = f(S_1,\ldots S_K)) = \omega^{\EA}$, it follows that $\Pr(Y=X_1'\oplus \cdots \oplus X_K') = \omega^{\EA}$, and that $\Pr(Y=X_1'\oplus \cdots \oplus  X_K'  \oplus 1) = 1-\omega^{\EA}$. Therefore, they obtain a converted channel with $Y = X_1' \oplus \cdots \oplus X_K' \oplus Z'$ where $Z'$ is the effective noise with $\Pr(Z'=0) = \omega^{\EA}$ and $\Pr(Z'=1) = 1-\omega^{\EA}$. A BSC code with TDMA can then be applied to this channel to achieve any rate tuple that satisfies $R_1+R_2+\cdots + R_K < 1-H_b(\omega^{\EA})$.
\end{proof}

\section{The CHSH strategy} \label{proof:CHSH}
Let us explain the CHSH strategy in the context of our channel conversion scheme. Recall that for each channel use, Transmitter $1$ knows $S_1\in \{0,1\}$ and Transmitter $2$ knows $S_2\in \{0,1\}$. 
The two transmitters share a pair of entangled qubits (namely, $Q_1Q_2$) in the state $\ket{\Phi^+}_{Q_1Q_2} = \frac{1}{\sqrt{2}}\big(\ket{0}_{Q_1}\!\ket{0}_{Q_2}+\ket{1}_{Q_1}\!\ket{1}_{Q_2}\big) = \frac{1}{\sqrt{2}}\big(\ket{+}_{Q_1}\!\ket{+}_{Q_2}+\ket{-}_{Q_1}\!\ket{-}_{Q_2}\big)$, where $\ket{+} \triangleq \frac{1}{\sqrt{2}}\big(\ket{0}+\ket{1}\big)$ and $\ket{-} \triangleq \frac{1}{\sqrt{2}}\big(\ket{0}-\ket{1}\big)$. 
For $\theta\in [0,2\pi)$, let ${\sf U}_\theta \triangleq \ket{+}\!\bra{+} + e^{i\theta} \ket{-} \! \bra{-}$ be a unitary operator. For each $k\in \{1,2\}$, Transmitter $k$ applies ${\sf U}_{\theta_k}$ on $Q_k$, ($\theta_k$ depends on $S_k$ and is specified later). The state of the qubits after the operations is,
\begin{align}
	\ket{\Psi}_{Q_1Q_2} &= \frac{1}{\sqrt{2}}\Big(\ket{+}_{Q_1}\!\ket{+}_{Q_2} + e^{i(\theta_1+\theta_2)}\ket{-}_{Q_1}\!\ket{-}_{Q_2}\Big) \\
	&= \frac{1}{2\sqrt{2}}\big( 1+e^{i(\theta_1+\theta_2)} \big)\big(\ket{0}_{Q_1}\!\ket{0}_{Q_2} +\ket{1}_{Q_1}\!\ket{1}_{Q_2} \big) \notag \\
	&~~~~+ \frac{1}{2\sqrt{2}}\big( 1-e^{i(\theta_1+\theta_2)} \big)\big(\ket{0}_{Q_1}\!\ket{1}_{Q_2} +\ket{1}_{Q_1}\!\ket{0}_{Q_2} \big)
\end{align}
For $k\in \{1,2\}$, Transmitter $k$ measures $Q_k$ in the computational basis and obtains $B_k \in \{0,1\}$ by mapping the state $\ket{x}_{Q_k} \to x,\forall x\in \{0,1\}$. Conditioned on $(S_1,S_2)$, the probability of $B_1\oplus B_2 = 0$ is then $2 \times |\frac{1}{2\sqrt{2}}(1+e^{i(\theta_1+\theta_2)})|^2 = \cos^2\big(\frac{\theta_1+\theta_2}{2}\big)$, and  the probability of $B_1\oplus B_2 = 1$ is $1-\cos^2\big(\frac{\theta_1+\theta_2}{2}\big) = \sin^2\big(\frac{\theta_1+\theta_2}{2}\big)$. Now, let $\theta_1 = 0$ if $S_1=0$; $\theta_1 = \pi/2$ if $S_1=1$. Let $\theta_2 = -\pi/4$ if $S_2=0$; $\theta_2=\pi/4$ if $S_2=1$. It can be then verified that $\Pr(B_1\oplus B_2 = 0\mid S_1S_2=0) = \Pr(B_1\oplus B_2 = 1\mid S_1S_2=1) = \cos^2(\pi/8)$. Thus, $\Pr(B_1\oplus B_2 \oplus S_1S_2=0) = \cos^2(\pi/8) = \frac{1}{2}+\frac{\sqrt{2}}{4}$. \hfil \qed

\bibliographystyle{IEEEtran}
\bibliography{../../bib_file/yy.bib}
\end{document}